\documentclass[final]{siamltex}
\usepackage{amssymb}
\usepackage{graphicx}
\usepackage{algorithm}
\usepackage{algpseudocode}
\usepackage{subfig}

\title{INACCESSIBILITY-INSIDE THEOREM FOR POINT IN
  POLYGON\thanks{'eligible for the best student paper award' A
    copy of this work can be found at http://arxiv.org/abs/1010.0552
    in the open arXiv of Cornell University Library.}}  

\author{Shriprakash Sinha$^{1,2}$ and Luca Nanetti$^{1}$ \thanks{($1$) Neuroimaging Center,
    UMCG, Antonious Deusinglaan 2, 9713 AW Groningen, The
    Netherlands. ($2$) ICT, EEMCS, Pattern Recognition and Bioinformatics
    Group, Mekelweg 4, 2628 CD Delft, The
    Netherlands. Questions, comments, or
    corrections to this document may be directed to the above email
    addresses. Email - shriprakash.sinha@gmail.com, l.nanetti@med.umcg.nl.}}

\begin{document}
\maketitle

\begin{abstract}
The manuscript presents a theoretical proof in conglomeration with new definitions on \emph{Inaccessibility} and
\emph{Inside} for a point $\mathcal{S}$ related to a simple or self
intersecting polygon $\mathcal{P}$. The proposed analytical solution
depicts a novel way of solving the point in polygon problem by
employing the properties of epigraphs and hypographs, explicitly. Contrary to the
ambiguous solutions given by the cross over for the simple and self
intersecting polygons and the solution of a point being multi-ply
inside a self intersecting polygon given by the winding number rule,
the current solution gives unambiguous and singular result for both
kinds of polygons. Finally, the current theoretical solution proves to
be mathematically correct for simple and self intersecting polygons. 
\end{abstract}

\begin{keywords} Student, Inaccessibility, Inside, Point, Polygon,
  Epigraph, Hypograph.\end{keywords}

\begin{AMS} 65D18, 68U05\end{AMS}

\pagestyle{myheadings}
\thispagestyle{plain}
\markboth{INACCESSIBILITY-INSIDE THEOREM FOR POINT IN POLYGON}{Shriprakash Sinha and Luca Nanetti}

\section{Introduction}\label{sec:intro}

Given a polygon $\mathcal{P}$ or the vertices of the polygon, say
$(x_{1},y_{1})$, $(x_{2},y_{2})$ ... $(x_{n},y_{n})$, it is desired to know whether a sample point $\mathcal{S}$
$(x_{0},y_{0})$ lies within $\mathcal{P}$. \emph{The status of a point
  $\mathcal{S}$ related to a polygon $\mathcal{P}$ being termed as
  inside needs to be defined correctly}. The definition is very
important in order to retreive unambiguous results for not only simple
but intersecting polygons also in a $2D$ Cartesian plane. In this
manuscript, new definition of \emph{inaccessibility} and \emph{inside}
has been proposed in order to accurately specify the meaning behind
the inclusion of a point within or without a polygon.\par  
\textbf{Cross Over} (\cite{Nordbeck:1967}, \cite{Manber:1989},
\cite{Foley:1990}, \cite{Haines:1994}) states that if a semi infinite
line drawn from $\mathcal{S}$ cuts the $\mathcal{P}$ odd number of
times, then the point is inside the polygon. Three issues arise in this
case, i.e. $\bullet$ depending on the \textbf{orientation} of the line from the
query point, odd or even values can be obtained, if the line passes
through vertices. This gives rise to ambiguous results for the same
point with different rays at different orientation. A prevalent
solution is the shifiting of the ray infinitesimally, but then again
solution may change drastically depending on the \textbf{direction of the
shift}. Even though this may be a rare case with negligible chance of
occurance, the issue persists, leading to ambiguous results. $\bullet$
A second issue is that of repeating the cross over multiple times
until the point lies inside the polygon. This leads to \textbf{non determinism}
as it is not known how many times the rays need to be shot to get an
affirmative answer, if ever it is conducted. $\bullet$ In case of
intersecting polygons, areas exist which give a different solution
than the winding number rule concept. \par
\textbf{Winding Number Rule} (\cite{Haines:1994}, \cite{Huang:1997}, \cite{Zalik:2001})
states that the number of times one loops
around $\mathcal{S}$ while traversing $\mathcal{P}$ before reaching
the starting point on the polygon shows whether the point is inside
the polygon or not. So a number $\ell$ greater than one can mean
that the point is $\ell$ times inside the polygon. This is an issue
because if a point lies inside a polygon once, it lies forever. Thus
$\ell > 1$ depicts the idea of \textbf{redundancy}. \par
As will be explained later in detail, the current solution looks at
these problems afflicting the status of point related to polygon from
a different perspective. The manuscript defines the concepts of
\emph{inaccessibility} and \emph{inside} of a polygon $\mathcal{P}$
while relating to the query point $\mathcal{S}$. The proposed solution
is motivated from \cite{Chen:1987} but has a \emph{major difference}
in using $\mathcal{S}$ as a reference point to draw a line chain
through it, that cuts the polygon at different intersection
points. The following section (\ref{sec:novel}) shed light on novel
algorithm explained with the assumptions involved. Next a theoretical
analysis of the solution is given in \ref{sec:proof}. A detailed
comparison with the crossover and the winding number with the proposed
algorithm is made in sections \ref{sec:covseh} and
\ref{sec:wnrvseh}. Finally, the conclusion is reached in section
\ref{sec:conclusion}.  For detailed analysis of the experimental results
and time complexity of the algorithm please see the appendix below or visit
http://arxiv.org/abs/1010.0552. \par
%
\section{A Novel Algorithm}\label{sec:novel}
The novel algorithm in simple terms can be described as follows. Given
the sample point ($x_{0}, y_{0}$), a horizontal line $y = y_{0}$ is
drawn through $\mathcal{S}$ to cut the $\mathcal{P}$ at $q$ locations
$\{$($x_{1}^{int}, y_{0}$), ... , ($x_{q}^{int}, y_{0}$) $\}$, thus
breaking the polygon into $q$ chains. \par
\begin{definition}
A \emph{chain} is a series of connected edges of the polygon whose starting and
ending points lie on the horizontal straight line that passes through $\mathcal{S}$. Mathematically, a
chain is a function $f_{c}$, with a closed domain defined by the
starting and ending points on the horizontal line passing through the
point of test $\mathcal{S}$, and a range that is the graph of the
currently under investigation connected edges of the polygon
. \label{def:chain}
\end{definition} \par
Each chain is then checked for whether its two endpoints contain the
test point between them; if not, the chain is discarded. Discarded
chains are termed as \emph{invalid chains} and those kept for further
consideration are referred to as \emph{valid chains}. The remaining
chains are then tested for intersection with a vertical line $x =
x_{0}$ through $\mathcal{S}$. The intersections found are sorted by
height, and paired up. If the test point is not between a pair, it is
outside. This criterion of containment is checked via the definitions
of affine sets and affine combination as follows: \par
\begin{definition}
A set $\mathcal{T} \subseteq \mathcal{R}^{n}$ is an \emph{affine set}, if for any
two points $x_{i}, x_{j} \in \mathcal{T} (j > i)$ and $\theta \in [0,1]$,
$\theta x_{i} + (1-\theta) x_{j} \in \mathcal{T}$. \label{def:affineset}
\end{definition} \par
\begin{definition}
An \emph{affine combination} of $x_{i}, x_{j} \in \mathcal{R}$ are a set of points of
the the form $\theta_{i} x_{i} + \theta_{j} x_{j}$, where $\theta_{i}
+ \theta_{j} = 1$. \label{def:affinecomb}
\end{definition} \par
These definitions and notations and a few others, are adopted from
\cite{Boyd:2004}. It is assumed that the vertices of the polygon
$\mathcal{P}$ are arranged in order of traversal, starting from one of
the vertices. The traversal order can be in any one direction. Another
assumption is that the edges are traversed only once. This is useful
in avoiding multiple loops that may occur in cases of intersecting
polygons. \par
If $\mathcal{S}$ lies out of the bounding box of the polygon, it is
considered outside $\mathcal{P}$ and no further processing is
done. Lastly, if the sample point is one of the vertices of the
polygon, then it is considered to be in the polygon. This final point
is assumed as the proposed solution would reach the same conclusion at
the expense of computational time. In the theoretical proof, it will
be shown that the assumption for implementation issue is correct. \par
As the algorithm is explained the concepts of epigraph and
hypographs will also be used for providing an analytically complete
elucidation of the generated solution. The definition of these are as
follows: \par
\begin{definition}
The \emph{epigraph} of a function (chain)
$f_{c}:\mathcal{R}^{n}\rightarrow\mathcal{R}$ is a set of points that
lie on or above the graph under consideration, such that $epi(f_{c}) =
\{(x,t) : x \in \mathcal{R}^{n}, t \in \mathcal{R}, f_{c}(x) \leq t\}$
is a subset of $\mathcal{R}^{n+1}$. \label{def:epigraph}
\end{definition} \par
\begin{definition}
The \emph{hypograph} of a function (chain)
$f_{c}:\mathcal{R}^{n}\rightarrow\mathcal{R}$ is a set of points that
lie on or below the graph under consideration, such that $hypo(f_{c}) = \{(x,t) :
x \in \mathcal{R}^{n}, t \in \mathcal{R}, f_{c}(x) \geq t\}$ is a
subset of $\mathcal{R}^{n+1}$. \label{def:hypograph}
\end{definition} \par
Finally, to decide if the point lies inside or is inaccessible with
respect to a polygon under consideration, the definition of nearest
chains would be needed. This definition is as follows: \par
\begin{definition}
Chains $\mathcal{C}_{i}$ and $\mathcal{C}_{j}$ are \emph{nearest
  valid chains} if:
\begin{itemize}
\item[$\bullet$] the $epi(f_{\mathcal{C}_{i}}) \subset epi(f_{\mathcal{C}_{u}})$
  $\forall u \in {1,...,q}$ chains below $\mathcal{S}$ such that
  $(x_{0},y_{0}) \in epi(f_{\mathcal{C}_{u}})$.
\item[$\bullet$] the $hypo(f_{\mathcal{C}_{j}}) \subset
  hypo(f_{\mathcal{C}_{v}})$ $\forall v \in {1,...,q}$ chains above
  $\mathcal{S}$ such that $(x_{0},y_{0}) \in
  hypo(f_{\mathcal{C}_{v}})$.
\end{itemize}
\label{def:nearestchain}
\end{definition} \par
%
\section{Inaccessibility-Inside Theorem}\label{sec:proof}
Given the new solution, it becomes imperative to prove the correctness
of the solution. This follows due to the fact that definitions like
the cross over and the winding number rule exist that state the
meaning of inside from different perspectives, thus giving
contradictory results. New definitions of \emph{inside} and
\emph{inaccessibility} of a point $\mathcal{S}$ related to polygon
$\mathcal{P}$ are proposed and a relation between inaccessibility and
inside is proved. \emph{This proof shows that consistent results can be
obtained if the meaning of the inaccessibility and inside of a polygon
related to a point is are framed correctly in an abstract sense}. \par
It must be noted that the points that lie on vertices and edges are
special cases and the definitions of inside and inaccessibility get
slightly modified. But this does not mean that the meaning of
inaccessibility and inside get twisted or modified from an abstract
sense. Two cases are presented, one that is a point lying on a vertex
and the other for the general case where it lies either on the edge or
anywhere else. \par
\subsection{Point on Vertex of Polygon}\label{sec:case_1}
The definitions of inaccessibility and inside are proposed for the
case of a point lying on a vertex. The essence of abstract meaning of
the same gets carried over to points not on vertex also but the
definitions are slgihtly modified. \par
\begin{definition}
The \emph{inaccessibility} $Inacc_{\mathcal{P}}(\mathcal{S})$ of a
point $\mathcal{S}$ related to a polygon $\mathcal{P}$, is the number
of valid chains that need to be \emph{broken}.
\[Inacc_{\mathcal{P}}(\mathcal{S}) = \left\{ 
\begin{array}{l l}
  \mathcal{N}, & \mbox{$\mathcal{N} \neq 0$ valid chains to be \emph{broken}}\\
  0, & \mbox{otherwise}\\ \end{array} \right. \]
\label{def:inacc_case_1}
\end{definition} \par
\begin{definition}
The status of a point $\mathcal{S}$ related to a polygon
$\mathcal{P}$, that is $Inside_{\mathcal{P}}(\mathcal{S})$, is the
existance of a chain $\mathcal{C}$ such that $\mathcal{S} \in
epi(f_{\mathcal{C}})$ or $\mathcal{S} \in hypo(f_{\mathcal{C}})$.
\[Inside_{\mathcal{P}}(\mathcal{S}) = \left\{ 
\begin{array}{l l}
  1, & \mbox{if $\mathcal{S} \in epi(f_{\mathcal{C}})$ or
    $\mathcal{S} \in hypo(f_{\mathcal{C}})$}\\
  0, & \mbox{otherwise}\\ \end{array} \right. \]
\label{def:inside_case_1}
\end{definition} \par
Based on these two definitions, two throrems need to be proved
regarding the relationship of inaccessibility of a point as well as
the status of a point whether it is inside with respect to the
polygon. \par
\begin{theorem}
A point $\mathcal{S}$ related to polygon $\mathcal{P}$ is inside as
well as inaccessible when:
$Inside_{\mathcal{P}}(\mathcal{S}) \in \{1\}$ iff
$Inacc_{\mathcal{P}}(\mathcal{S}) \in \{\mathcal{N}\}$
\label{theo:case_1_a}
\end{theorem}
\begin{proof}
(a) If $Inacc_{\mathcal{P}}(\mathcal{S}) \in \{\mathcal{N}\}$ then
$Inside_{\mathcal{P}}(\mathcal{S}) \in \{1\}$ \\\par
Given that $Inacc_{\mathcal{P}}(\mathcal{S}) = \mathcal{N}$, there
exists $\mathcal{N}$ valid chains that need to be \emph{broken} according to
definition \ref{def:inacc_case_1}. It is known that a chain is valid
when either its epigraph or hypograph contains $\mathcal{S}$. This
existance of $\mathcal{N}$ valid chains imply that $\mathcal{S} \in
\{epi(f_{\mathcal{C}_{k}}), hypo(f_{\mathcal{C}_{k}})\}$ $\forall k
\in \{1,...,\mathcal{N}\}$. But this is the definition of stauts of
$\mathcal{S}$ related to $\mathcal{P}$,
i.e. $Inside_{\mathcal{P}}(\mathcal{S}) = 1$ or
$Inside_{\mathcal{P}}(\mathcal{S}) \in \{1\}$. \\\par
(b) If $Inside_{\mathcal{P}}(\mathcal{S}) \in \{1\}$ then
$Inacc_{\mathcal{P}}(\mathcal{S}) \in \{\mathcal{N}\}$ \\\par
Given $Inside_{\mathcal{P}}(\mathcal{S}) = 1$ implies that
$\mathcal{S} \in \{epi(f_{\mathcal{C}}), hypo(f_{\mathcal{C}})\}$ for a
chain $\mathcal{C}$ in $f$. Thus chain $\mathcal{C}$ is a valid chain,
as it contains the point $\mathcal{S}$. In order for $\mathcal{S}$ to
be inaccessible, there must exist atleast $1$ vaild chain in
$\mathcal{P}$ that needs to be \emph{broken}. Since $\mathcal{C}$ is one such
chain and the only chain that contains $\mathcal{S}$, the
inaccessibility order of $\mathcal{S}$ related to $\mathcal{P}$ in
$Inacc_{\mathcal{P}}(\mathcal{S}) = 1$ or
$Inacc_{\mathcal{P}}(\mathcal{S}) \in \{1\}$. \par
If $\mathcal{S}$ is a vertex such that it is an intersection point of
two or more lines of a polygon, then all chains that have their
epigraph or hypograph contain $\mathcal{S}$, are valid. Since it
requires $\mathcal{N}$ (if $\mathcal{N}$ is the number of valid
chains) chains to be \emph{broken}.
\end{proof}
Cases need to be shown pictorially to get a feel of what the theorem
is suggesting about. Figure \ref{fig:case_1_a} shows three different
polygons with $\mathcal{S}$ as the point under consideration. The
polygon in figure \ref{fig:case_1_a}.(A) has four chains that contain
$\mathcal{S}$ namely (a) $\mathcal{S}$TU (b) UV$\mathcal{S}$ (c)
  $\mathcal{S}$WX and (d) XY$\mathcal{S}$, which are valid. Thus by
  theorem \ref{theo:case_1_a}, $Inside_{\mathcal{P}}(\mathcal{S}) = 1$
  and $Inacc_{\mathcal{P}}(\mathcal{S}) = 4$. Thus $\mathcal{S}$ lies
  inside the polygon. Similarly, for figure \ref{fig:case_1_a}.(B)
  there is one chain $\mathcal{S}$TU$\mathcal{S}$ which is valid as it
  contains the point $\mathcal{S}$. Thus
  $Inside_{\mathcal{P}}(\mathcal{S}) = 1$ and
  $Inacc_{\mathcal{P}}(\mathcal{S}) = 4$. For the case of figure
  \ref{fig:case_1_a}.(C) there exists two chains that contain
  $\mathcal{S}$ i.e. (a) $\mathcal{S}$TU$\mathcal{S}$ and
  $\mathcal{S}$VW$\mathcal{S}$ which are valid. So
  $Inside_{\mathcal{P}}(\mathcal{S}) = 1$ and
  $Inacc_{\mathcal{P}}(\mathcal{S}) = 2$. \par
Note that since this holds true always when $\mathcal{S}$ lies on the
vertex of a polygon, it is obvious and correct to assume that the
point is in the polygon by first checking if $\mathcal{S}$ is any one
of the vertices in the polygon. This helps to avoid the
implementation hurdle of checking the theorem. But again it is
stressed that first the point needs to be checked against vertices of
the polygon, in order to know if they belong to $\mathcal{P}$. \par
\begin{theorem}
A point $\mathcal{S}$ related to polygon $\mathcal{P}$ is not inside as
well as not inaccessible when:
$Inside_{\mathcal{P}}(\mathcal{S}) \in \{0\}$ iff
$Inacc_{\mathcal{P}}(\mathcal{S}) \in \{0\}$
\label{theo:case_1_b}
\end{theorem}
\begin{proof}
(a) If $Inacc_{\mathcal{P}}(\mathcal{S}) \in \{0\}$ then
$Inside_{\mathcal{P}}(\mathcal{S}) \in \{0\}$ \\\par
Given that $Inacc_{\mathcal{P}}(\mathcal{S}) = 0$, there
exists no valid chains that need to be \emph{broken} according to
definition \ref{def:inacc_case_1}. This means that $\mathcal{S} \notin
\{epi(f_{\mathcal{C}_{k}}), hypo(f_{\mathcal{C}_{k}})\}$ $\forall k$
valid chains in $\mathcal{P}$. Since no chain exists whoes epigraph or
hypograph contains $\mathcal{S}$, the status of $\mathcal{S}$ related
tp $\mathcal{P}$ is $Inside_{\mathcal{P}}(\mathcal{S}) = 0$ or
$Inside_{\mathcal{P}}(\mathcal{S}) \in \{0\}$. \\\par
(b) If $Inside_{\mathcal{P}}(\mathcal{S}) \in \{0\}$ then
$Inacc_{\mathcal{P}}(\mathcal{S}) \in \{0\}$ \\\par
Given $Inside_{\mathcal{P}}(\mathcal{S}) = 0$ implies that
$\mathcal{S} \notin \{epi(f_{\mathcal{C}_{k}}),
hypo(f_{\mathcal{C}_{k}})\}$ $\forall k$ chains in $\mathcal{P}$. This
means no valid chains exist in $\mathcal{P}$ that need to be
\emph{broken}. Thus the inaccessibility of $\mathcal{S}$ related to
$\mathcal{P}$ is zero, i.e. $Inacc_{\mathcal{P}}(\mathcal{S}) \in
\{0\}$, which is the desired result.
\end{proof}
Cases for theorem \ref{theo:case_1_b} are simple and depicted in
figure \ref{fig:case_1_b}. Figure \ref{fig:case_1_b} shows two different
polygons with $\mathcal{S}$ as the point under consideration. The
polygon in figure \ref{fig:case_1_b}.(A) has four chains that do not contain
$\mathcal{S}$ namely (a) RTU (b) UVR (c) RWX and (d) XYR, which are
invalid. Thus by theorem \ref{theo:case_1_b},
$Inside_{\mathcal{P}}(\mathcal{S}) = 0$ and
$Inacc_{\mathcal{P}}(\mathcal{S}) = 0$. Thus $\mathcal{S}$ lies
outside the polygon. Similarly, for figure \ref{fig:case_1_b}.(B)
there exists two chains that do not contain $\mathcal{S}$ i.e. (a)
RTUR and RVWR, which are invalid. So
$Inside_{\mathcal{P}}(\mathcal{S}) = 0$ and
$Inacc_{\mathcal{P}}(\mathcal{S}) = 0$. \par
\begin{figure}
\centering
\subfloat[$\mathcal{S}$ as a vertex point]{\label{fig:case_1_a}\includegraphics[width=.3\textwidth]{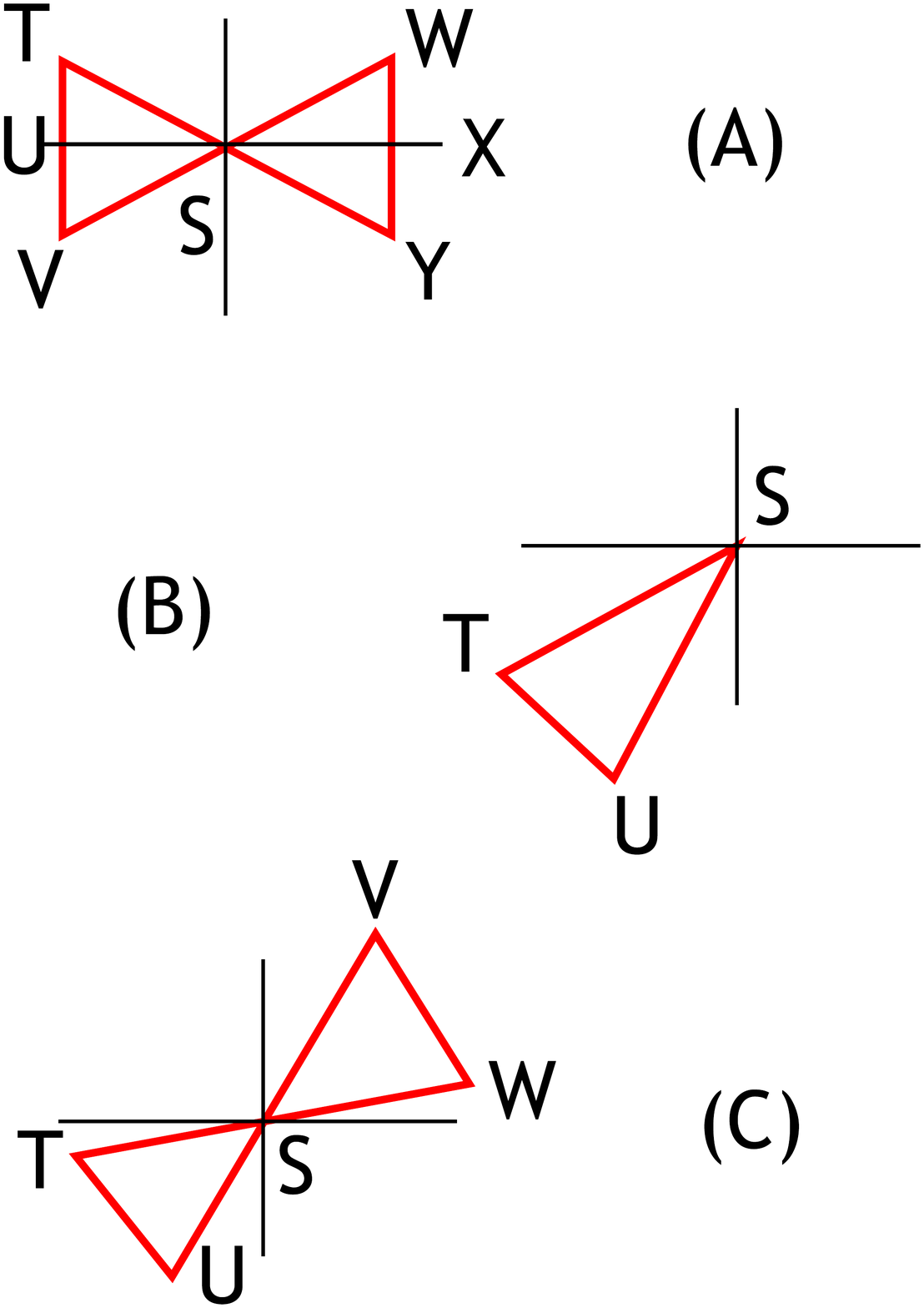}}
\subfloat[$\mathcal{S}$ not on the vertex]{\label{fig:case_1_b}\includegraphics[width=.3\textwidth]{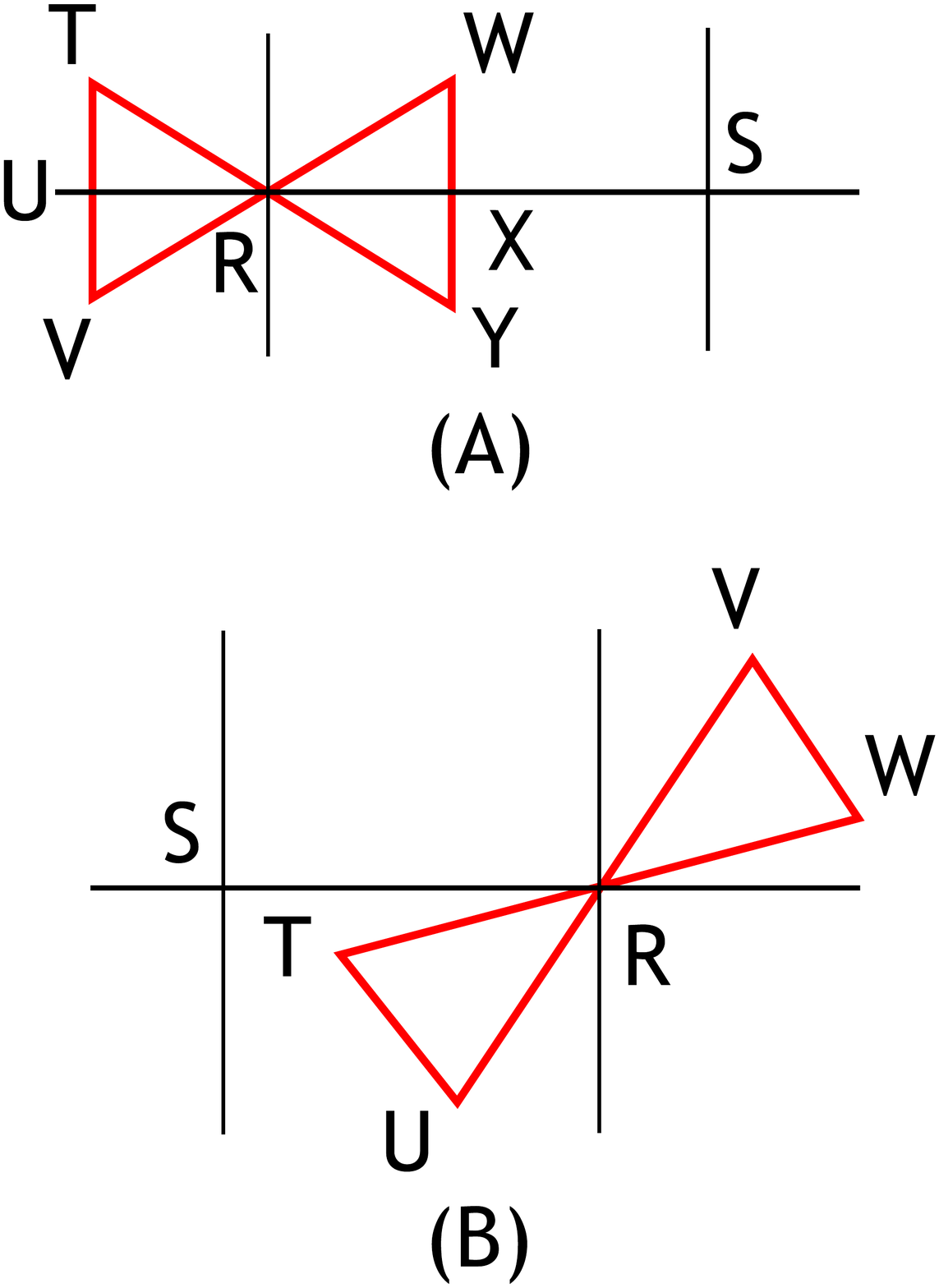}}
\caption{Polygons with locations of the point $\mathcal{S}$.}
\label{fig:case_1}
\end{figure}
\subsection{Point not on Vertex of Polygon}\label{sec:case_2}
Now for the general case of point in polygon, the definition of
inaccessibility and inside evolve slightly while preserving the
abstract essence of the idea. Again the theorems will be proved which
find a relation between when a point is inaccessibile as well as
inside the polygon. \par
\begin{definition}
The \emph{inaccessibility} $Inacc_{\mathcal{P}}(\mathcal{S})$ of a
point $\mathcal{S}$ related to a polygon $\mathcal{P}$, is the number
of valid chains that need to be \emph{broken and/or ignored}.
\[Inacc_{\mathcal{P}}(\mathcal{S}) = \left\{ 
\begin{array}{l l}
  1, & \mbox{a pair of chains need to be \emph{broken}}\\
  \mathcal{N}, & \mbox{$\mathcal{N} \neq 0$ pairs of chains to be \emph{ignored}}\\
  1+\mathcal{N}, & \mbox{a pair to be \emph{broken} and}\\ 
  & \mbox{$\mathcal{N}$ pairs to be \emph{ignored}}\\ \end{array} \right. \]
\label{def:inacc_case_2}
\end{definition}
\begin{definition}
The status of a point $\mathcal{S}$ related to a polygon
$\mathcal{P}$, that is $Inside_{\mathcal{P}}(\mathcal{S})$, is the
existance of a pair of chains $\mathcal{C}_{i}$ and $\mathcal{C}_{j}$
such that $\mathcal{S} \in epi(f_{\mathcal{C}_{i}})$ and $\mathcal{S}
\in hypo(f_{\mathcal{C}_{j}})$.
\[Inside_{\mathcal{P}}(\mathcal{S}) = \left\{ 
\begin{array}{l l}
  1, & \mbox{pairs of chains $\mathcal{C}_{i}$ and $\mathcal{C}_{j}$,}\\
    & \mbox{s.t. $\mathcal{S} \in epi(f_{\mathcal{C}_{i}})$ and
      $\mathcal{S} \in hypo(f_{\mathcal{C}_{j}})$}\\
  0, & \mbox{otherwise}\\ \end{array} \right. \]
\label{def:inside_case_2}
\end{definition}
Again the relation between inaccessibility and inside of a polygon is
proved via two theorems. The theorems are as follows: \par
\begin{theorem}
A point $\mathcal{S}$ related to a polygon $\mathcal{P}$ is inside as
well as inaccessible when:
$Inside_{\mathcal{P}}(\mathcal{S}) \in \{1\}$ iff
$Inacc_{\mathcal{P}}(\mathcal{S}) \in \{1, 1+\mathcal{N}\}$
\label{theo:case_2_a}
\end{theorem}
\begin{proof}
(a) If $Inacc_{\mathcal{P}}(\mathcal{S}) \in \{1, 1+\mathcal{N}\}$
then $Inside_{\mathcal{P}}(\mathcal{S}) \in \{1\}$.\\\par
Given $Inacc_{\mathcal{P}}(\mathcal{S}) \in \{1, 1 + \mathcal{N}\}$
implies that there exist a pair of valid chain in $\mathcal{P}$ that
need to be \emph{broken} and/or $\mathcal{N}$ pairs of valid chains
that need to be ignored. A valid chain by definition is one whoes
epigraph or hypograph contains $\mathcal{S}$. Taking the general case
of $1+\mathcal{N}$ (if $\mathcal{N} = 0$, $1+\mathcal{N}$ collapses to
$1$), there are $2\times(1+\mathcal{N})$ valid chains such that half
lie above/on $\mathcal{S}$ and the rest half lie below/on
$\mathcal{S}$. If a vertical line passing through $x = x_{0}$ is drawn
such that it cuts the valid chains and $\mathcal{S}$, then the chains
can be sorted according to the value of intersection points in $x =
x_{0}$. Let $\mathcal{C}_{1}$ ,...,
$\mathcal{C}_{2\times(1+\mathcal{N})-1}$,
$\mathcal{C}_{2\times(1+\mathcal{N})}$ be the sorted order of chains
from bottom to top. Taking consecutive pairs of these valid chains
i.e. ($\mathcal{C}_{1}$, $\mathcal{C}_{2}$), ($\mathcal{C}_{3}$,
$\mathcal{C}_{4}$), ..., ($\mathcal{C}_{i}$, $\mathcal{C}_{i+1}$),
..., ($\mathcal{C}_{2\times(1+\mathcal{N})-1}$,
$\mathcal{C}_{2\times(1+\mathcal{N})}$), it is easy to know whether
$\mathcal{S}$ is an affine combination of ($x_{0},
y_{\mathcal{C}_{k}}^{int}$) and ($x_{0},
y_{\mathcal{C}_{k+1}}^{int}$), $\forall \in
\{1, 3, 5,... , 2\times(1+\mathcal{N}) - 1\}$. Here
$y_{\mathcal{C}_{k}}^{int}$ and $y_{\mathcal{C}_{k+1}}^{int}$ are the
intersection points on the chains $k$ and $k+1$ due to the line $x =
x_{0}$. \par
Since it is known that atleast one pair of chains need to be
\emph{broken}, a pair of points ($x_{0}, y_{\mathcal{C}_{i}}^{int}$)
and ($x_{0}, y_{\mathcal{C}_{i+1}}^{int}$) exists for which
$\mathcal{S}$ ($x_{0}, y_{0}$) $\equiv$ ($x_{0},
\theta y_{\mathcal{C}_{i}}^{int}+(1-\theta)
y_{\mathcal{C}_{i+1}}^{int}$) for $0 \leq \theta \leq 1$. This implies
that $\mathcal{C}_{i}$ is the nearest chain below/on $\mathcal{S}$ and
$\mathcal{C}_{i+1}$ is the nearest chain above/on $\mathcal{S}$,
otherwise $\mathcal{S}$ won't be an affine combination of ($x_{0},
y_{\mathcal{C}_{i}}^{int}$) and ($x_{0},
y_{\mathcal{C}_{i+1}}^{int}$). For the rest of the $\mathcal{N}$
pairs, since $\mathcal{S}$ is not an affine combination of ($x_{0},
y_{\mathcal{C}_{k}}^{int}$) and ($x_{0}, y_{\mathcal{C}_{k+1}}^{int}$)
$\forall k \in \{1, 3, ..., 2\times(1+\mathcal{N})-1\} - \{i\}$, these
$\mathcal{N}$ pairs of chains can be \emph{ignored} from further
processing or consideration. Since these nearest chains
$\mathcal{C}_{i}$ and $\mathcal{C}_{i+1}$ are also valid, their
epigraph and hypograph contain $\mathcal{S}$, respectively. This
existence of a pair of a valid chains which has to be \emph{broken}
such that $\mathcal{S} \in epi(f_{\mathcal{C}_{i}})$ and $\mathcal{S}
\in hypo(f_{\mathcal{C}_{i+1}})$ implies
$Inside_{\mathcal{P}}(\mathcal{S}) = 1$, the status of $\mathcal{S}$
related to $\mathcal{P}$. \\\par
(b) If $Inside_{\mathcal{P}}(\mathcal{S}) \in \{1\}$ then
$Inacc_{\mathcal{P}}(\mathcal{S}) \in \{1, 1+\mathcal{N}\}$ \\\par
Let $\mathcal{P}$ be a polygon such that
$Inside_{\mathcal{P}}(\mathcal{S}) = 1$ implies there exists a
pair of chains $\mathcal{C}_{i}$ and $\mathcal{C}_{j}$ such that
$\mathcal{S} \in epi(f_{\mathcal{C}_{i}})$ and $\mathcal{S} \in
hypo(f_{\mathcal{C}_{j}})$. Given only these two chains, it is evident
that both of them are nearest chins to $\mathcal{S}$. Let the starting
and ending points of $\mathcal{C}_{i}$ and $\mathcal{C}_{j}$ be
$\{(x_{\mathcal{C}_{i_{s}}}^{int}, y_{0}),
(x_{\mathcal{C}_{i_{e}}}^{int}, y_{0})\}$ and
$\{(x_{\mathcal{C}_{j_{s}}}^{int}, y_{0}),
(x_{\mathcal{C}_{j_{e}}}^{int}, y_{0})\}$, respectively. If a vertical 
line $x = x_{0}$ is drawn through ($x_{0}, y_{0}$) it would interest
the chains $\mathcal{C}_{i}$ and $\mathcal{C}_{j}$ at ($x_{0},
y_{\mathcal{C}_{i}}^{int}$) and ($x_{0}, y_{\mathcal{C}_{j}}^{int}$),
respectively. Since ($x_{0}, y_{\mathcal{C}_{i}}^{int}$) lies below
($x_{0}, y_{0}$) and ($x_{0}, y_{\mathcal{C}_{j}}^{int}$) lies above
($x_{0}, y_{0}$), $\mathcal{S}$ is an affine combination of ($x_{0},
y_{\mathcal{C}_{i}}^{int}$) and ($x_{0},
y_{\mathcal{C}_{j}}^{int}$). Thus $\mathcal{S}$ lies between
$\mathcal{C}_{i}$ and $\mathcal{C}_{j}$. Now, if an end point of
$\mathcal{C}_{i}$ is joined with an end point of $\mathcal{C}_{j}$ and
another end point of the former joined to the remaining end point of
the latter, then a closed loop is formed such that traversing once
from any one point, lead to the same point in the end. Let this loop
be $\mathcal{P}^{'}$. \par
As long as the winding number of $\mathcal{P}^{'}$ is around
$\mathcal{S}$ is same as that of $\mathcal{P}$ around $\mathcal{S}$,
$\mathcal{P}^{'}$ can be deformed into $\mathcal{P}$, by Hopf's degree
theorem \cite{Needham:1999}. Since $\mathcal{P}$ is formed from the
two necessary chains $\mathcal{C}_{i}$ and $\mathcal{C}_{j}$ which are
valid, a pair exists in polygon $\mathcal{P}$ that needs to be
\emph{broken}. Thus the minimum inaccessibility of $\mathcal{S}$
related to $\mathcal{P}$ is $Inacc_{\mathcal{P}}(\mathcal{S}) = 1$. If
there exists extra pairs of valid chains, then they would be
\emph{ignored} from consideration, while checking for the affine
combination criteria of $\mathcal{S}$ with respect to the sorted pairs
of chains on $x = x_{0}$. If $\mathcal{N}$ is the minimum number of
pairs of valid chains that are \emph{ignored}, then the
inaccessibility of $\mathcal{S}$ related to $\mathcal{P}$ is
$Inaccp_{\mathcal{P}}(\mathcal{S}) = 1+\mathcal{N}$. Thus
$Inacc_{\mathcal{P}}(\mathcal{S}) = \{1, 1+\mathcal{N}\}$.
\end{proof}
Now, the cases of theorem \ref{theo:case_2_a} are presented with
visual representations in figures \ref{fig:case_2_a_1} and
\ref{fig:case_2_a_2}. In figure \ref{fig:case_2_a_1}.(A), $3$ chains
exist of which two are valid. The valid chains are (a) UV and (b)
WU. The chain VW is invalid. Thus $Inside_{\mathcal{P}}(\mathcal{S}) =
1$ and $Inacc_{\mathcal{P}}(\mathcal{S}) = 1$. In figure
\ref{fig:case_2_a_1}.(B) the point lies on the edge and the polygon
can be divided into $3$ chains of which only two contain
$\mathcal{S}$. These valid chains are (a) SU and (b) VS. Thus
$Inside_{\mathcal{P}}(\mathcal{S}) = 1$ and
$Inacc_{\mathcal{P}}(\mathcal{S}) = 1$. In part (C) of figure
\ref{fig:case_2_a_1}, the horizontal line cuts through two edges and
touches two vertices. Thus there exists $4$ chains of which two are
valid, namely (a) UV and (b) XU which contain $\mathcal{S}$. Thus
$Inside_{\mathcal{P}}(\mathcal{S}) = 1$ and
$Inacc_{\mathcal{P}}(\mathcal{S}) = 1$. Lastly, in part (D) of the
same figure, $\mathcal{S}$ lies on an edge and the horizontal line
passes through a vertex. In this case, there exist $5$ chains of which
two contain the point of test and are thus valid. They are (a) VW and
(b) YU. Thus $Inside_{\mathcal{P}}(\mathcal{S}) = 1$ and
$Inacc_{\mathcal{P}}(\mathcal{S}) = 1$. \par
Next, in figure \ref{fig:case_2_a_2}.(A) $8$ chains exist namely, (a)
YZ (b) ZW (c) WB (d) BU (e) UV (f) VA (g) AX and (h) XY, of which $6$
are valid except UV and XY. These two do not contain
$\mathcal{S}$. Now, the pair that needs to be \emph{broken} is YZ and
ZW, while the pairs (WB, AX) and (BU, VA) are to be \emph{ignored}
from consideration. Thus, $Inside_{\mathcal{P}}(\mathcal{S}) = 1$ and
$Inacc_{\mathcal{P}}(\mathcal{S}) = 1+2$ (i.e one pair requires to be
broken and two need to be ignored from consideration). Finally, in
figure \ref{fig:case_2_a_2}.(B) two chains exist of which both are
valid, i.e. (a) UV and (b) VWWU. Thus
$Inside_{\mathcal{P}}(\mathcal{S}) = 1$ and
$Inacc_{\mathcal{P}}(\mathcal{S}) = 1$. \par
\begin{figure}
\centering
\subfloat[$\mathcal{S}$ not on the vertex]{\label{fig:case_2_a_1}\includegraphics[width=.3\textwidth]{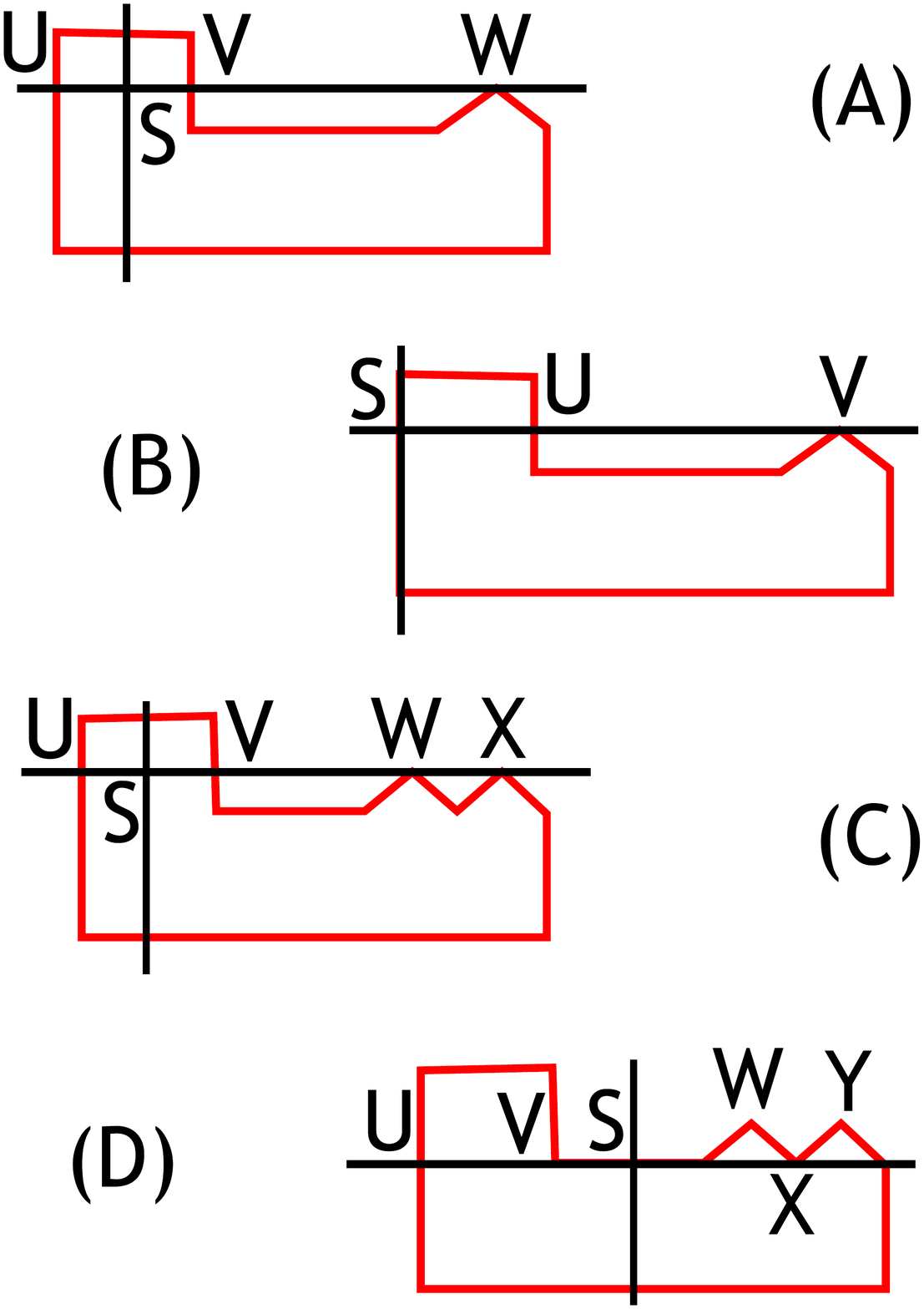}}
\subfloat[$\mathcal{S}$ not on the vertex]{\label{fig:case_2_a_2}\includegraphics[width=.3\textwidth]{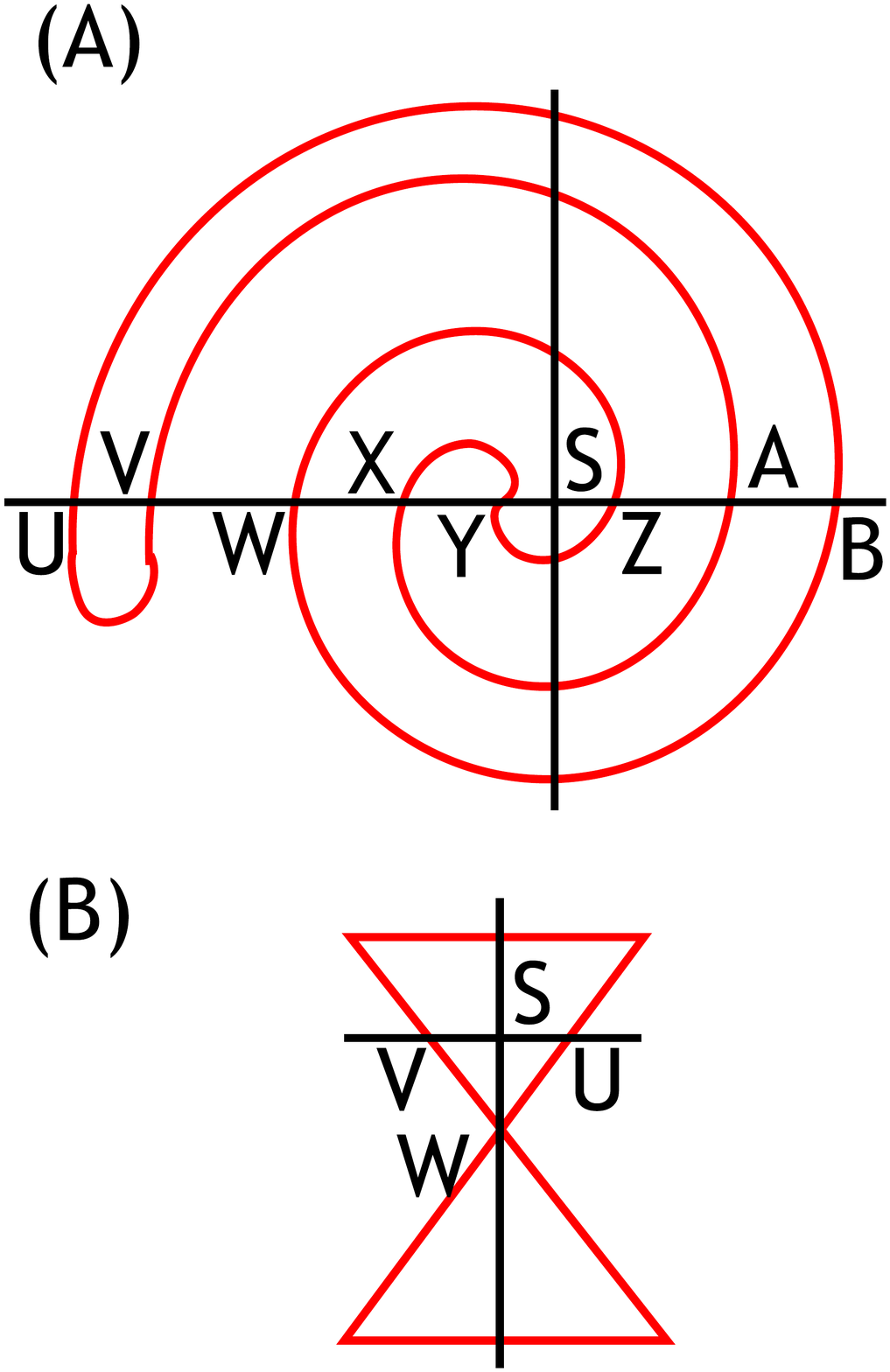}}
\caption{Polygons with locations of the point $\mathcal{S}$.}
\label{fig:case_2}
\end{figure}
\begin{theorem}
A point $\mathcal{S}$ related to a polygon $\mathcal{P}$ is not inside as
well as inaccessible when:
$Inside_{\mathcal{P}}(\mathcal{S}) \in \{0\}$ iff
$Inacc_{\mathcal{P}}(\mathcal{S}) \in \{\mathcal{N}\}$
\label{theo:case_2_b}
\end{theorem}
\begin{proof}
(a) If $Inacc_{\mathcal{P}}(\mathcal{S}) \in \{\mathcal{N}\}$ then
$Inside_{\mathcal{P}}(\mathcal{S}) \in \{0\}$. \\\par
Given $Inacc_{\mathcal{P}}(\mathcal{S}) = \mathcal{N}$ implies pairs
of valid chains need to be \emph{ignored}. Again, a valid chain by
definition is one whoes epigraph and hypograph contain
$\mathcal{S}$. If a vertical line passing through $x = x_{0}$ is drawn
such that it cuts the valid chains and $\mathcal{S}$, then the chains
can be sorted according to the value of intersection points in $x =
x_{0}$. Let $\mathcal{C}_{1}, ...,\mathcal{C}_{2\times\mathcal{N}}$ be
the sorted order of chains from bottom to top. Taking consecutive
pairs of these valid chains, i.e. ($\mathcal{C}_{1}$,
$\mathcal{C}_{2}$), ($\mathcal{C}_{3}$, $\mathcal{C}_{4}$),...,
($\mathcal{C}_{i}$, $\mathcal{C}_{j}$), ...,
($\mathcal{C}_{2\times\mathcal{N}-1}$,
$\mathcal{C}_{2\times\mathcal{N}}$), it is easy to know whether
$\mathcal{S}$ is an affine combination of ($x_{0}$,
$y_{\mathcal{C}_{k}}^{int}$) and ($x_{0}$,
$y_{\mathcal{C}_{k+1}}^{int}$) $\forall k \in \{1, 3,
...,2\times\mathcal{N}-1\}$. Since $\mathcal{N}$ pairs need to be
\emph{ignored}, it is evident that $\mathcal{S}$ is not an affine
combination of any of the above pairs. This suggests that there does
not exist a pair such that $\mathcal{S} \in epi(f_{\mathcal{C}_{k}})$
and $\mathcal{S} \in hypo(f_{\mathcal{C}_{k+1}})$, that can be
\emph{broken}. Since no such pair exists, the status of $\mathcal{S}$
related to $\mathcal{P}$ is $Inside_{\mathcal{P}}(\mathcal{S}) = 0$,
which is the desired result. \\\par
(b) If $Inside_{\mathcal{P}}(\mathcal{S}) \in \{0\}$ then
$Inacc_{\mathcal{P}}(\mathcal{S}) \in \{\mathcal{N}\}$. \\\par
Given $Inside_{\mathcal{P}}(\mathcal{S}) = 0$ implies there does not
exist a pair of chains $\mathcal{C}_{i}$ and $\mathcal{C}_{j}$ such
that $\mathcal{S} \in epi(f_{\mathcal{C}_{i}})$ and $\mathcal{S} \in
hypo(f_{\mathcal{C}_{j}})$. Thus it is difficult to proceed with the
proof. Instead, by proving its contrapositive the above statement will
hold. Thus if $Inacc_{\mathcal{P}}(\mathcal{S}) \notin
\{\mathcal{N}\}$ then $Inside_{\mathcal{P}}(\mathcal{S}) \notin
\{0\}$. Since $Inacc_{\mathcal{P}}(\mathcal{S}) \notin
\{\mathcal{N}\}$, it implies that $Inacc_{\mathcal{P}}(\mathcal{S}) \in
\{1, 1+\mathcal{N}\}$. It has been proved in part (a) above that if
$Inacc_{\mathcal{P}}(\mathcal{S}) \in \{1, 1+\mathcal{N}\}$ then
$Inside_{\mathcal{P}}(\mathcal{S}) \in \{1\}$. But
$Inside_{\mathcal{P}}(\mathcal{S}) \in \{1\}$ also means that
$Inside_{\mathcal{P}}(\mathcal{S}) \notin \{0\}$. Thus
$Inacc_{\mathcal{P}}(\mathcal{S}) \notin \{\mathcal{N}\}$ implies that
$Inside_{\mathcal{P}}(\mathcal{S}) \notin \{0\}$. Since the
contrapositive holds, so does the original statement. 
\end{proof}
Finally, two cases for theorem \ref{theo:case_2_b} are depicted in
figure\ref{fig:case_2_b}. In figure \ref{fig:case_2_b}.(A) $8$ chains
exist namely, (a) UV (b) VA (c) AX (d) XY (e) YZ (f) ZW (g) WB and (h)
BU. Of these, $3$ pairs exist, all of which are valid chains but need
to be \emph{ignored} as non of them enclose $\mathcal{S}$ as an affine
combination of two intersection points on $x = x_{0}$. These pairs are
(WB, AX), (XY, ZW) and (VA, BU). Thus
$Inside_{\mathcal{P}}(\mathcal{S}) = 0$ and
$Inacc_{\mathcal{P}}(\mathcal{S}) = 3$. For the case of intersecting
polygon in figure \ref{fig:case_2_b} $6$ chains exist namely, (a) ZU
(b) UY (c) YW (d) WX (e) VX and (f) XZ, of which WY and XZ are not
valid. Remaining pairs of valid chains need to be \emph{ignored} and
thus $Inside_{\mathcal{P}}(\mathcal{S}) = 0$ and
$Inacc_{\mathcal{P}}(\mathcal{S}) = 2$. \par
\begin{figure}[!t]
\centering
\includegraphics[width=4cm,height=6cm]{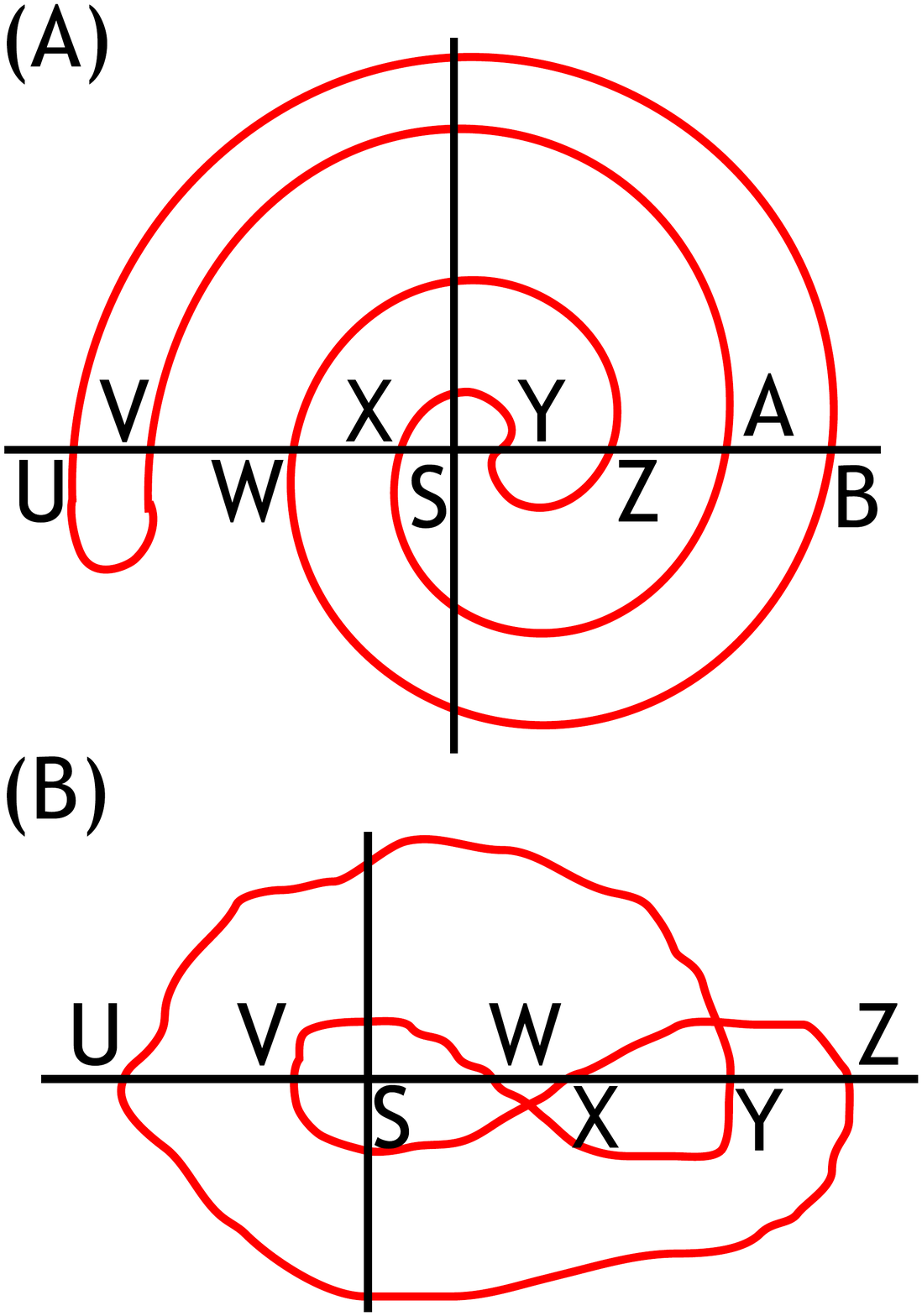}
\caption{Polygons with $\mathcal{S}$ outside.}
\label{fig:case_2_b}
\end{figure}
For the next few sections, let EH (epi/hypo-graph method) denote the proposed method.\par
%
\section{Crossover vs EH}\label{sec:covseh}
Crossover (Cr) states that a line drawn from a point $\mathcal{S}$ in
a direction, if it cuts the polygon $\mathcal{P}$ odd number of times,
implies that $\mathcal{S}$ is inside $\mathcal{P}$, i.e. \par
\[Inside_{\mathcal{P}}^{CR}(\mathcal{S}) = \left\{ 
\begin{array}{l l}
  1, & \mbox{odd intersections}\\
  0, & \mbox{even intersections}\\ \end{array} \right. \]
For the case of vertices, the problem is solved by shifting the line
infinitesimally. Two issues arise when a line is shot from
$\mathcal{S}$ and it intersects a vertex. ($1$) There can be two
solutions, if the line is not shifted slightly. ($2$) If the crossover
has to be repeated several times until it finds an odd number of
intersections, then it is a nondeterministic problem, in case the line
is shot randomly. \par
By ($1$) ambiguity arises on the way a ray or line is shot from
$\mathcal{S}$ and by ($2$) nondeterminism arises due to repeatedness
because of the line being shot randomly. The following figures will
illustrate these issues in detail. In contrast to the Cr, by checking
through theorems \ref{theo:case_1_a} and \ref{theo:case_2_a}, the EH
method can easily determine if $\mathcal{S}$ lies in $\mathcal{P}$ or
not, deterministically. This is because whichever way a line is
drawn through $\mathcal{S}$, if it cuts the polygon, then it will
dismember $\mathcal{P}$ into a finite number of countable chains. If
it doesn't cut the polygon and $\mathcal{S}$ is a vertex, then also
there exist atleast one chain that contains $\mathcal{S}$. Searching
for these valid chains and then locating which of those need to be
\emph{broken} is deterministic. \par
The figure \ref{fig:covseh} shows the different cases under
consideration for the comparison of CR and EH method for the same
point of investigation. In figure \ref{fig:covseh}.(A), if the a
horizontal line is drawn to the left of the $\mathcal{S}$, then it
intersects at two points U and V and if it drawn to the right, it
intersects at the point W. According to CR, when the line is drawn to
the right of the $\mathcal{S}$, then $\mathcal{S}$ is inside the
polygon. If the line is drawn to the left of $\mathcal{S}$, then the
point is outside the polygon. This is definitely a case of
ambiguity. Also, the outcome of the CR depends on the direction of the
ray that is shot from $\mathcal{S}$. This makes the outcome of the
test nondeterministic in the sense that it is not known which ray
would give the correct result, if the rays are shot randomly. \par
The EH method overcomes this problem by segmenting the polygon into
finitely countable chains. The searching for an affine combination
of valid chains that may contain $\mathcal{S}$ is deterministic as
there are only limited number of chains available for checking. Thus
the outcome of the EH method is final and deterministic. If two
perpendicular rays with its intersection point as $\mathcal{S}$ are
drawn at a different orientation, thus intersecting the polygon at
different places, even then by rotating the oriented axis and the polygon to
horizontal vertical frame, the solution remains the same. Thus
randomness of the rays do not affect the outcome of the point
inclusion test for $\mathcal{S}$. For part (A) in the figure
\ref{fig:covseh}, by CR $Inside_{\mathcal{P}}^{CR}(\mathcal{S}) =
(0,1)$ depending on the number of intersections that is $(2,1)$. By
EH method, $Inside_{\mathcal{P}}^{EH}(\mathcal{S}) = 1$ and
$Inacc_{\mathcal{P}}^{EH}(\mathcal{S}) = 3$ by theorem
\ref{theo:case_1_a}. It must be noted that the inaccessibility of the
point related to the polygon may change but the status of
$\mathcal{S}$ related to $\mathcal{P}$ captured by the definition of
$Inside$ will not change if the point is inside the polygon. \par
Similarly, for the part (B) and (C) in figure \ref{fig:covseh}, by CR
the $Inside_{\mathcal{P}}^{CR}(\mathcal{S}) = (1,0)$ depending on the
number of intersections based on the direction of the ray which is
$(1,2)$. Finally, in figure \ref{fig:covseh}.(D), for point
$\mathcal{S}_{1}$ four valid chains exist namely, (a) VW (b) XU (c)
U$\mathcal{S}_{2}$ and (d) $\mathcal{S}_{2}$V, none of which need to
be \emph{broken} or \emph{ignored}. Thus by theorem
\ref{theo:case_2_a} $Inside_{\mathcal{P}}^{EH}(\mathcal{S}_{1}) = 0$
and $Inacc_{\mathcal{P}}^{EH}(\mathcal{S}_{1}) = 2$. By CR the outcome
of the inclusion test changes, that is
$Inside_{\mathcal{P}}^{CR}(\mathcal{S}_{1}) = (0,1)$ depending on the
intersections obtained by the direction of the ray that is
$(2,3)$. For the point $\mathcal{S}_{2}$, two valid chains exist
namely (a) $\mathcal{S}_{2}$V and (b) VWXU$\mathcal{S}_{2}$. Thus by
theorem \ref{theo:case_1_a},
$Inside_{\mathcal{P}}^{EH}(\mathcal{S}_{2}) = 1$ and
$Inacc_{\mathcal{P}}^{EH}(\mathcal{S}_{2}) = 2$. By CR
$Inside_{\mathcal{S}_{2}}^{EH}(\mathcal{S}_{2}) = 0$ as the number of
intersections is $4$. \par
\section{Winding Number Rule vs EH}\label{sec:wnrvseh}
The winding number rule (WNR) states that the number of times one
loops round the point $\mathcal{S}$ before reaching the starting point
on polygon $\mathcal{P}$, decides the number of times whether
$\mathcal{S}$ lies inside $\mathcal{P}$ or not. Thus \par
\[Inside_{\mathcal{P}}^{WNR}(\mathcal{S}) = \left\{ 
\begin{array}{l l}
  \mathcal{N}, & \mbox{$\mathcal{N}$ loops around $\mathcal{S}$}\\
  0, & \mbox{zero loops around $\mathcal{S}$}\\ \end{array} \right. \]
An analogy of a prison wall shown in figure \ref{fig:wnrvseh} is taken into
account to the explain the differences. Figure \ref{fig:wnrvseh}.(A)
is the intial structure of the prison and then the final structure is
shown in the part (B) of the same figure. Initially, via the WNR,
$\mathcal{S}_{1}$ was lying outside and $\mathcal{S}_{2}$ inside the
prison wall. Same is the verdict by the new method. Next a portion of
the prison wall is extended and the final structure looks like that in
figure \ref{fig:wnrvseh}.(B) . Note that $\mathcal{S}_{1}$ and
$\mathcal{S}_{2}$ are still outside the new prison via the new
definition, as the areas in which $\mathcal{S}_{1}$ and
$\mathcal{S}_{2}$ lie, are not reachable from prison's
perspective. This is because two pairs of walls have to be
\emph{ignored and not broken} in each case. From this point of view
both $\mathcal{S}_{1}$ and $\mathcal{S}_{2}$ are outside
$\mathcal{P}$, in figure \ref{fig:wnrvseh}.(B). Also, even though WNR
= 2 for $\mathcal{S}_{2}$ in new prison in part (B) of the same
figure, implies the point lies twice inside, it does not make
sense. It can be stated that if a point lies inside once, then it lies
forever. There does not arise the idea of point lying inside
$\mathcal{N}$ times. Thus a point lying inside $\mathcal{N}$ times, is
the same as point lying inside once. If it does not lie inside, then
it won't lie forever. In this way the new definitions and the
accompanying theorem are definitive in producing a concrete answer via
means of epigraph-hypograph method. \par
\begin{figure}
\centering
\subfloat[CR and EH ]{\label{fig:covseh}\includegraphics[width=.3\textwidth]{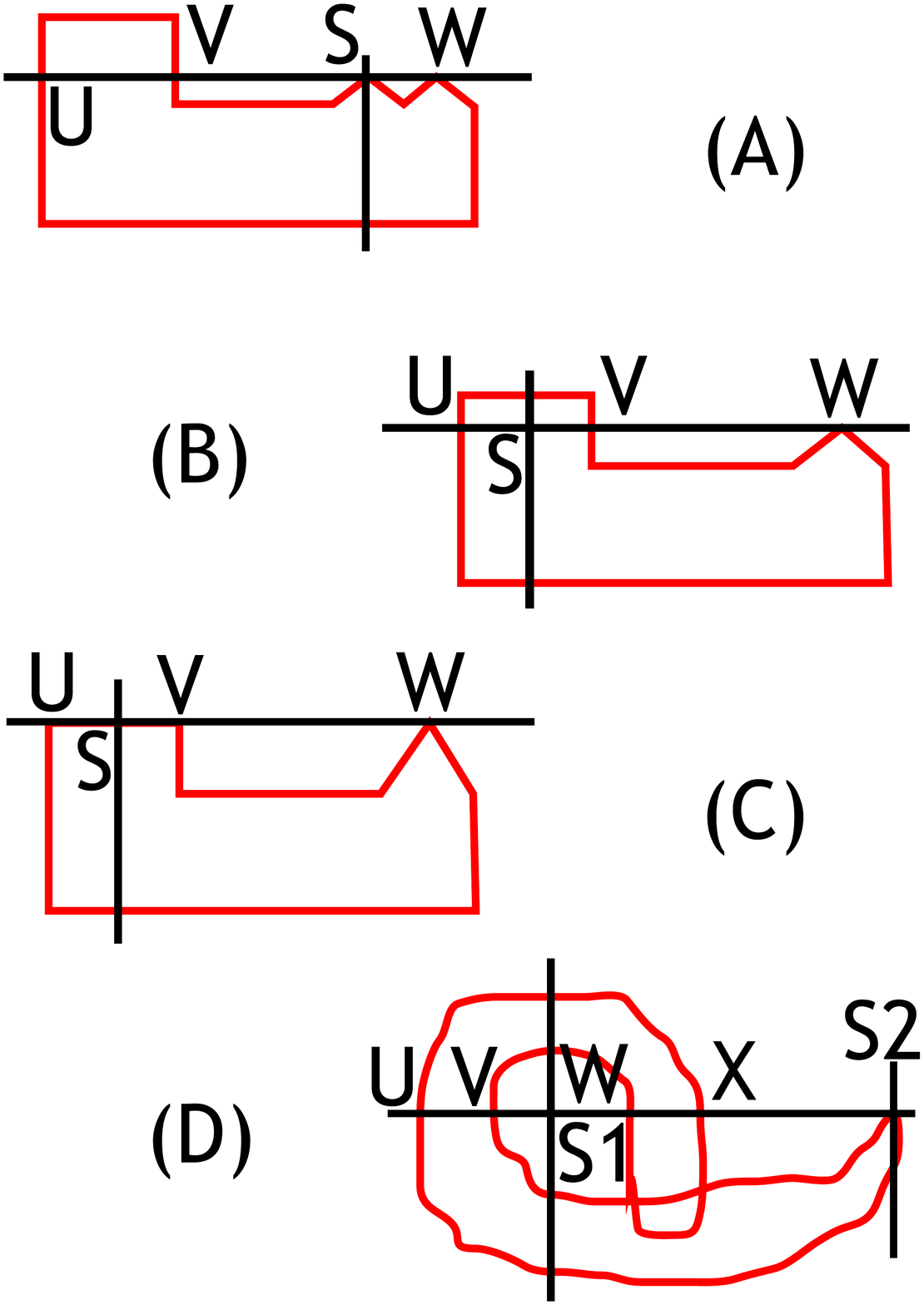}}
\subfloat[WNR and EH]{\label{fig:wnrvseh}\includegraphics[width=.3\textwidth]{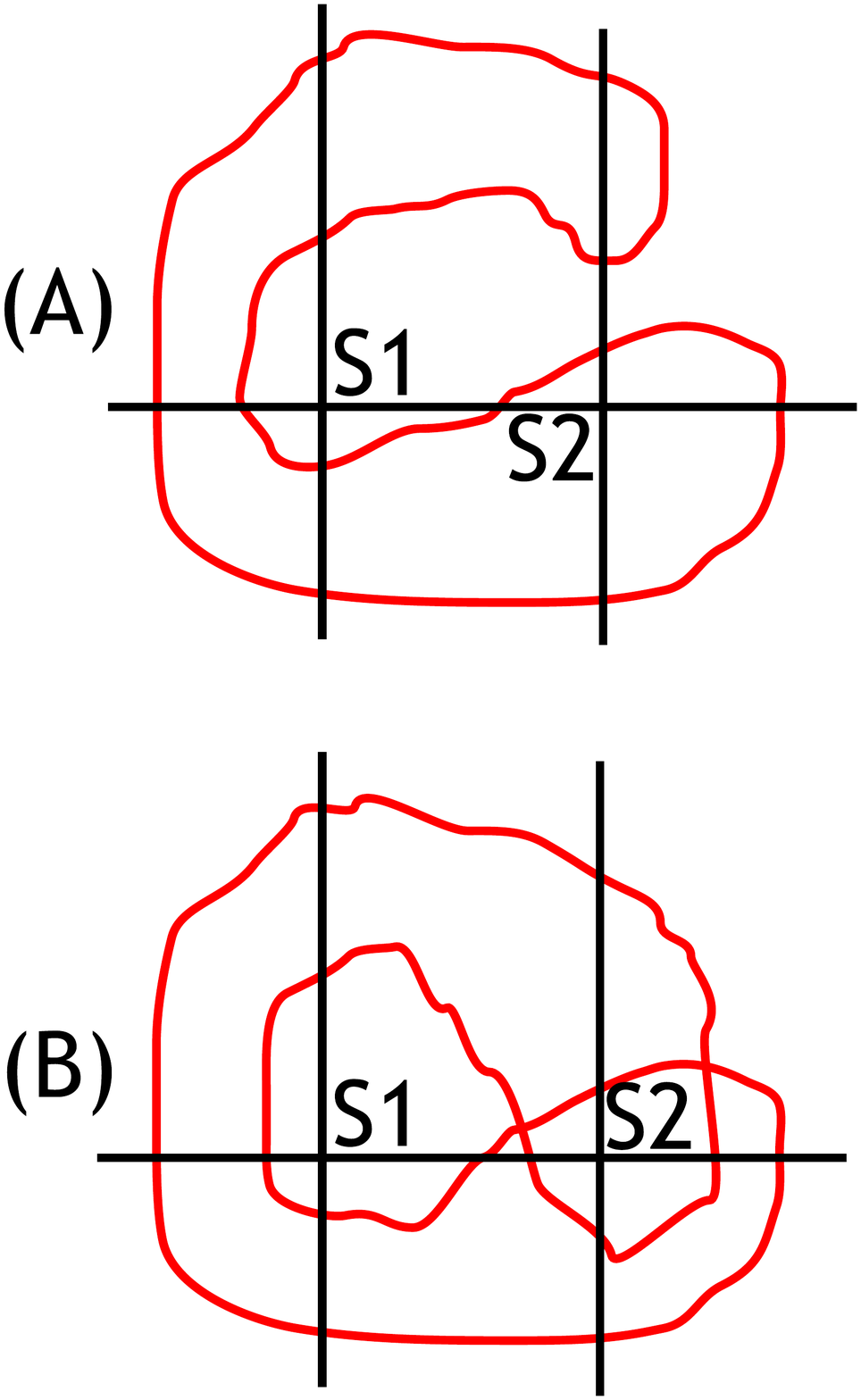}}
\caption{Cases to compare Crossover (CR), Winding Number Rule (WNR)
  and the proposed method (EH).}
\label{fig:crvswnrvseh}
\end{figure}
\section{Conclusion}\label{sec:conclusion}
A theoretically reliable and an analytically correct solution to the point in
polygon problem is proposed. The proof for the same is given by building
a relationship between two new concepts of inaccessibility as well as
inside. It is proved that the solution is good interpretation of
inside for both the simple and self intersecting polygons, which forms
the crux of the manuscript. \par 
%

\bibliographystyle{ieeetr}

\newpage
\section*{Appendix}

\section{Algorithm Implementation}\label{sec:alg_impl}
The implemented algorithm will be illustrated with each step explained
with a pictorial representation of result of the execution of that
particular step. Examples include closed, intersecting and non
intersecting polygons. Figures \ref{fig:in} and \ref{fig:out} show the
polygons with the sample point being tested at different locations. \par
\begin{figure}[!t]
\centering
\includegraphics[width=8.5cm,height=10cm]{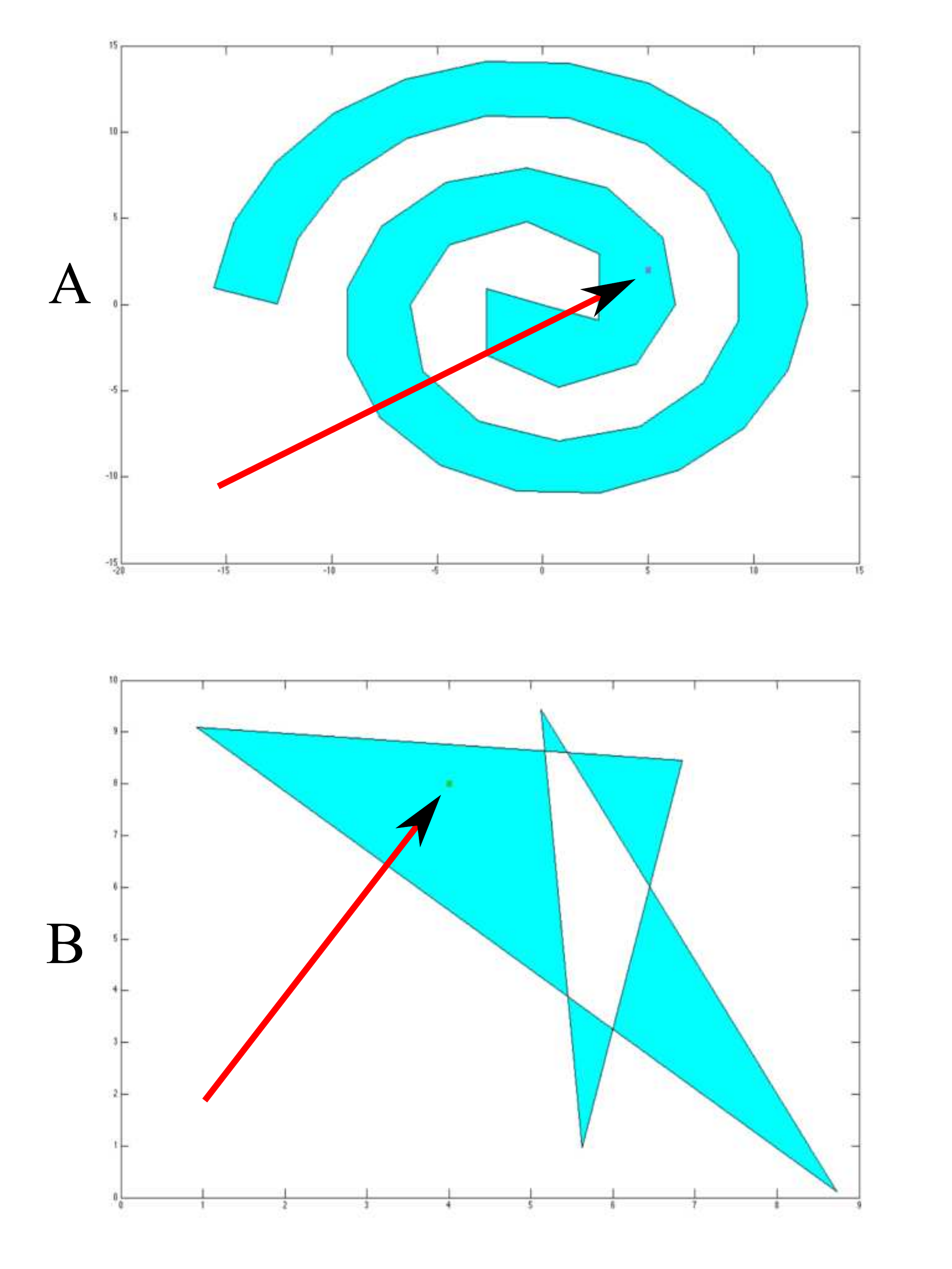}
\caption{Query point location in (A) closed and (B) intersecting polygon.}
\label{fig:in}
\end{figure}
\begin{figure}[!t]
\centering
\includegraphics[width=8.5cm,height=10cm]{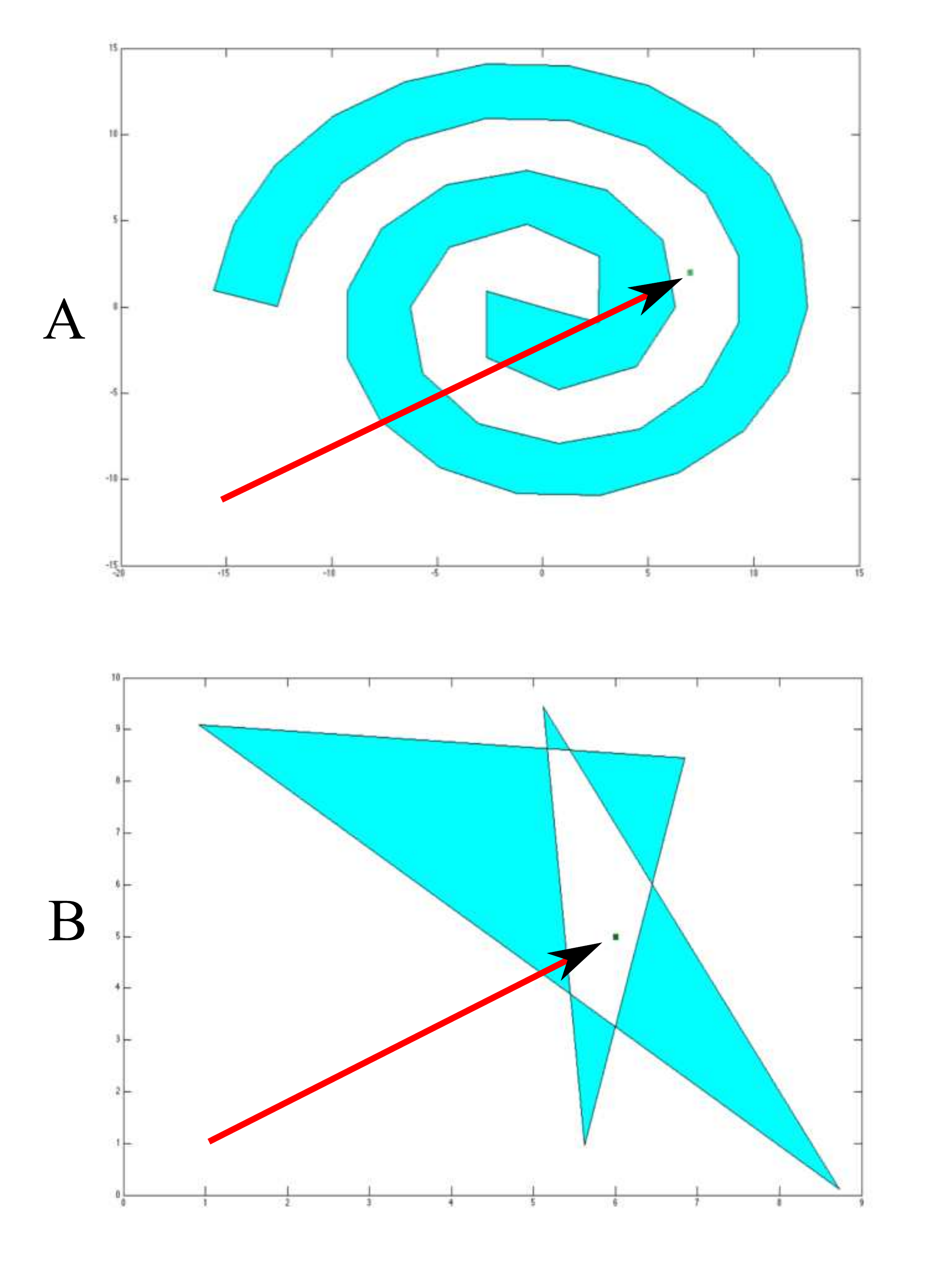}
\caption{Query point location in (A) closed and (B) intersecting polygon.}
\label{fig:out}
\end{figure}
\subsection{Intersecting the $\mathcal{P}$ }\label{sec:intersect_poly}
$\mathcal{P}$ is an ordered series of vertices $(x_{1},y_{1})$,
$(x_{2},y_{2})$,..., $(x_{n},y_{n})$, starting from
$(x_{1},y_{1})$ that defines the polygon. Given $\mathcal{P}$, the
first step is to draw a straight line through the sample point
$\mathcal{S}$ $(x_{0},y_{0})$, such that it intersects the polygon at
certain points. For simplicity, the horizontal line $y = y_{0}$ is
considered without loss of generality. \par 
The intersection points are obtained via computing the values of the
$x_{int}$ coordinates between the $y = y_{0}$ and the straight line,
extending from $(x_{i},y_{i})$ to $(x_{i+1},y_{i+1})$. The slope and
the constant of the former is $0$ and $y_{0}$, and that of latter is
$m_{i} = \frac{y_{i+1} - y_{i}}{x_{i+1} - x_{i}}$ and $c_{i} = y_{i} -
m_{i} x_{i}$. Here $i$ and $i+1$ are consecutive points on
$\mathcal{P}$. Solving the algebraic equation between two straight
line gives: \par
\begin{equation}
x^{int} = \frac{y_{0} - c_{i}}{m_{i}}
\label{eq:x_int}
\end{equation}
Once an intersecting point with coordinate $(x^{int},y_{0})$ is
obtained, the algorithm checks if this point lies in between
$(x_{i},y_{i})$ and $(x_{i+1},y_{i+1})$, \emph{on the line}. This is
achieved via affine combination property in definition
\ref{def:affinecomb}.  Three different cases arise depending on the
slope of the line: \par
\begin{enumerate}
\item $m_{i} = \pm \infty$ : If the edge is a vertical line, the point
  $(x^{int},y_{0})$ lies on the line. This is because $x^{int} = x_{i}
  = x_{i+1}$.
\item $m_{i} = 0$ : If edge is a horizontal line, the intersection point
  $(x^{int},y_{0})$ is considered to lie outside the range of $(x_{i},y_{0})$
  and $(x_{i+1},y_{0})$. This is because, if it lies within
  $(x_{i},y_{0})$ and $(x_{i+1},y_{0})$, then there are infinitely
  many points that could be considered. To save from randomly
  selecting any point within the given range, it is thus considered
  that the point $(x^{int},y_{0})$ outside the range. Also, if $(x_{0},
  y_{0})$ lie on a horizontal edge, it is still considered outside the
  range for further processing.
\item $m_{i} \in \mathcal{R}-\{0,\pm \infty\}$ : Finally, this being
  the simplest case, it is easy to compute whether $(x^{int},y_{0})$
  lie on the line between the given points using definition
  \ref{def:affinecomb}.
\end{enumerate}
Figures \ref{fig:step_1_in} and \ref{fig:step_1_out} show the
horizontal line $y = y_{0}$ passing through $\mathcal{S}$ and
intersecting the polygon at different edges. The sample point is
indicated via the red arrow while the vertices at the intersection
of $y = y_{0}$ with edges of the polygon are pointed to by the blue
arrows. Once the vertices of the intersecting points are computed,
they are added to the list of pre existing vertices of the
polygon under consideration. The newly created intersecting vertices
are appended in such a manner that the traversal order is not
affected. \par
It may happen that some of the newly added vertices with
their $y$ coordinate as $y_{0}$ are very close to other pre existing
vertices in the tolerance range of $\pm 10^{-5}$ or smaller. The
algorithm removes the vertices from the list that lie in such close range, but
keeps the newly added vertices with coordinates
$(x^{int},y_{0})$. Two reasons arise for executing this step:\par 
\begin{enumerate}
\item To avoid further computations that may involve floating point
  precision of the order smaller than or equal to $\pm 10^{-5}$.
\item To retain newly appended vertices with coordinates
  $(x^{int},y_{0})$ that will be later used for searching chains
  whose epigraph or hypograph may contain $\mathcal{S}$. 
\end{enumerate}
It must be noted that this insertion of vertices does not
affect the solution to the PiP on paper, but may affect solution
practically on a computer due to the floating point representation and
operations on it. Also, note that the tolerance range is not the
issue investigated in this work. \par
\begin{figure}[!t]
\centering
\includegraphics[width=8.5cm,height=10cm]{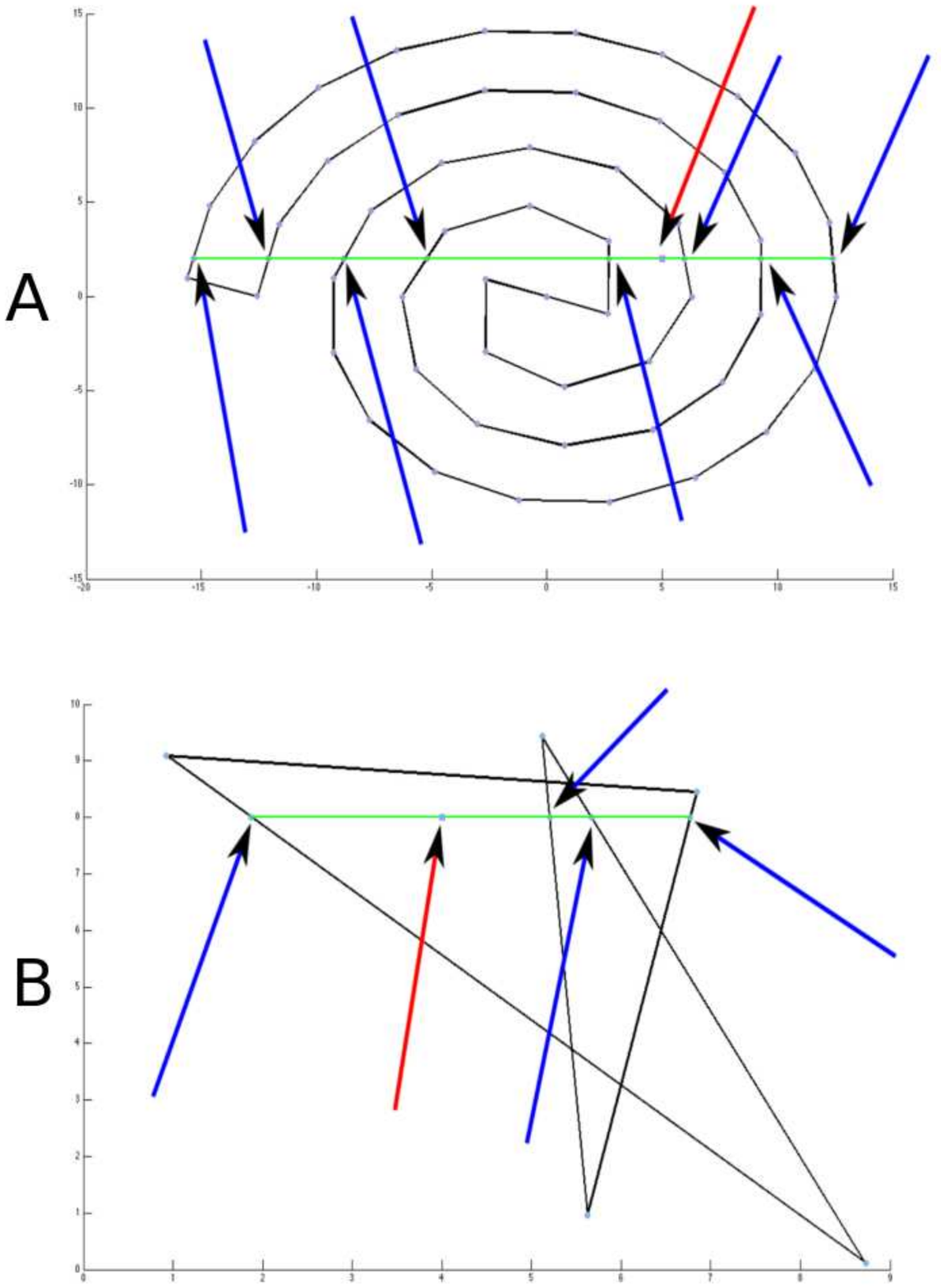}
\caption{Finding interesecting points on the polygon. Sample point
  indicated via red arrow lies in (A) closed and (B) intersecting
  polygon. Blue arrows point to newly found intersection points. The
  green line depicts $y = y_{0}$.} 
\label{fig:step_1_in}
\end{figure}
\begin{figure}[!t]
\centering
\includegraphics[width=8.5cm,height=10cm]{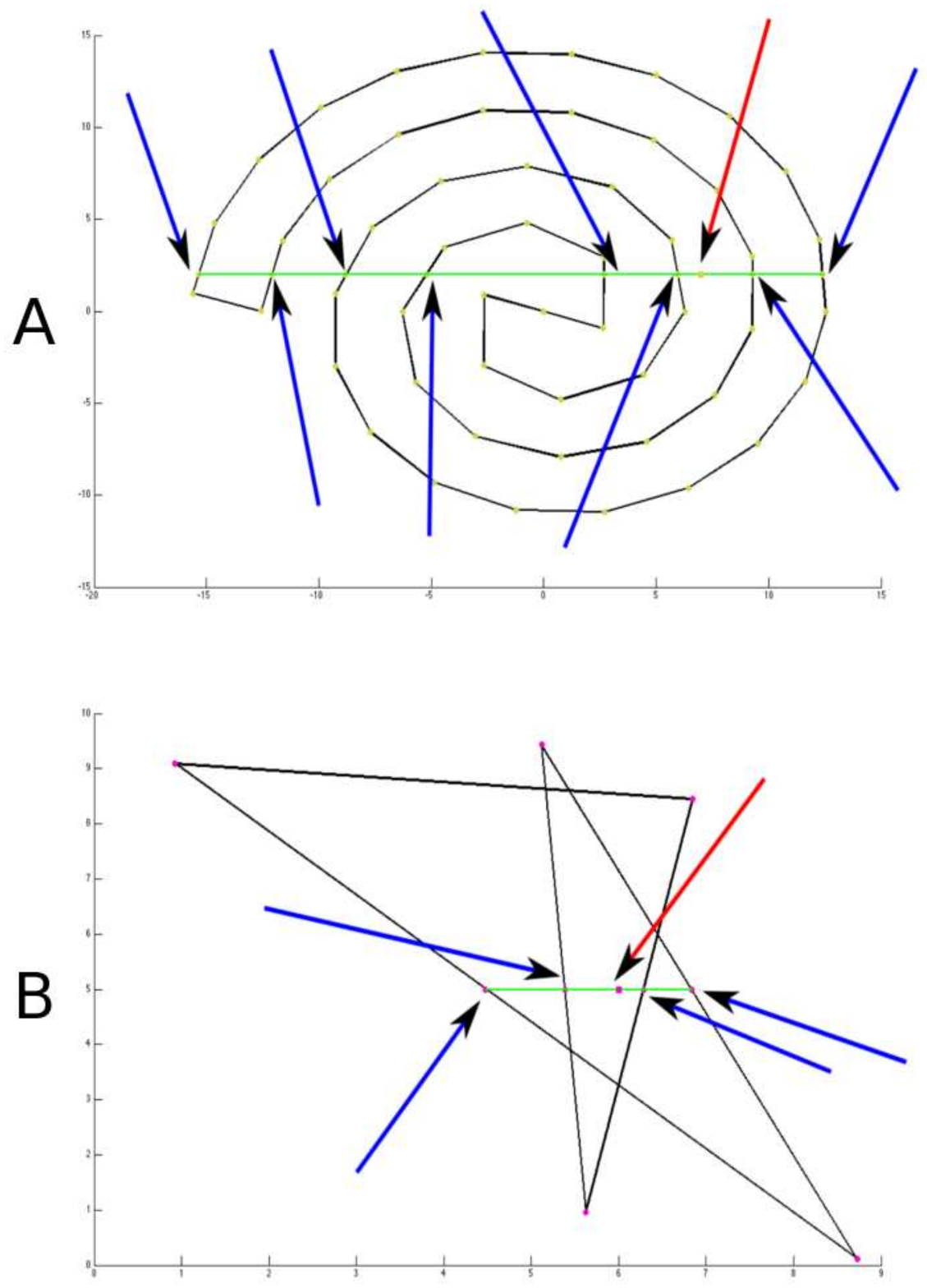}
\caption{Finding intersecting points on the polygon. Sample point
  indicated via red arrow lies out of (A) closed and (B) intersecting
  polygon. Blue arrows point to the newly found intersection
  points. The green line depicts $y = y_{0}$.}
\label{fig:step_1_out}
\end{figure}
%
\subsection{Decomposition of Polygon into Valid and Invalid
  Chains}\label{sec:decompose}
The new vertices with coordinates
$(x_{j}^{int},y_{0})$ were ($j \in \{1,...,m\}$) and the sample point
$(x_{0},y_{0})$ become form the basis for the next steps. From
definition \ref{def:epigraph} it is known that a point belongs to the
epigraph (hypograph) if it lies on or above (below) the graph under
consideration. \par
To use the mentioned properties, the polygon $\mathcal{P}$ needs to be
decomposed into the chains. These chains would then be tested
for convexity or concavity. The algorithm achieves this in the
following way. A vertex on $\mathcal{P}$, especially with the $y$
coordinate being $y_{0}$ is picked up as the starting vertex. A traversal
order is chosen randomly and is followed until the starting point is
reached again. \par
Vertices lying on a path between any two consecutive $(x_{i}^{int},y_{0})$
and $(x_{j}^{int},y_{0})$ points on $\mathcal{P}$, including the
intersecting points themselves, form a chain. Thus as the traversal
is done from one intersecting vertex to the other, the polygon gets
decomposed into chains, that lie either above or below the
horizontal line $y = y_{0}$. \par
Next the chains are tested for validity using the definitions in
\ref{def:epigraph} and \ref{def:hypograph}. In
non mathematical terms, if the starting and ending vertices of a
chain are on either side of the $\mathcal{S}$ on the line $y =
y_{0}$ and not on the same side with respect to the  sample point,
then the chain is a valid one. The valid chains are stored with
their starting and ending vertices along with coordinates of points on
the chain. \par
Figures \ref{fig:step_2_in} and \ref{fig:step_2_out} show the valid
and invalid chains pointed by the green and blue arrows
respectively. The solid and dotted lines also demarcate the
chains. The solid lines represent the valid chains and the dotted
lines represent the invalid chains. The red arrow indicates the
sample point's location in each of the figures. \par
\begin{figure}[!t]
\centering
\includegraphics[width=8.5cm,height=10cm]{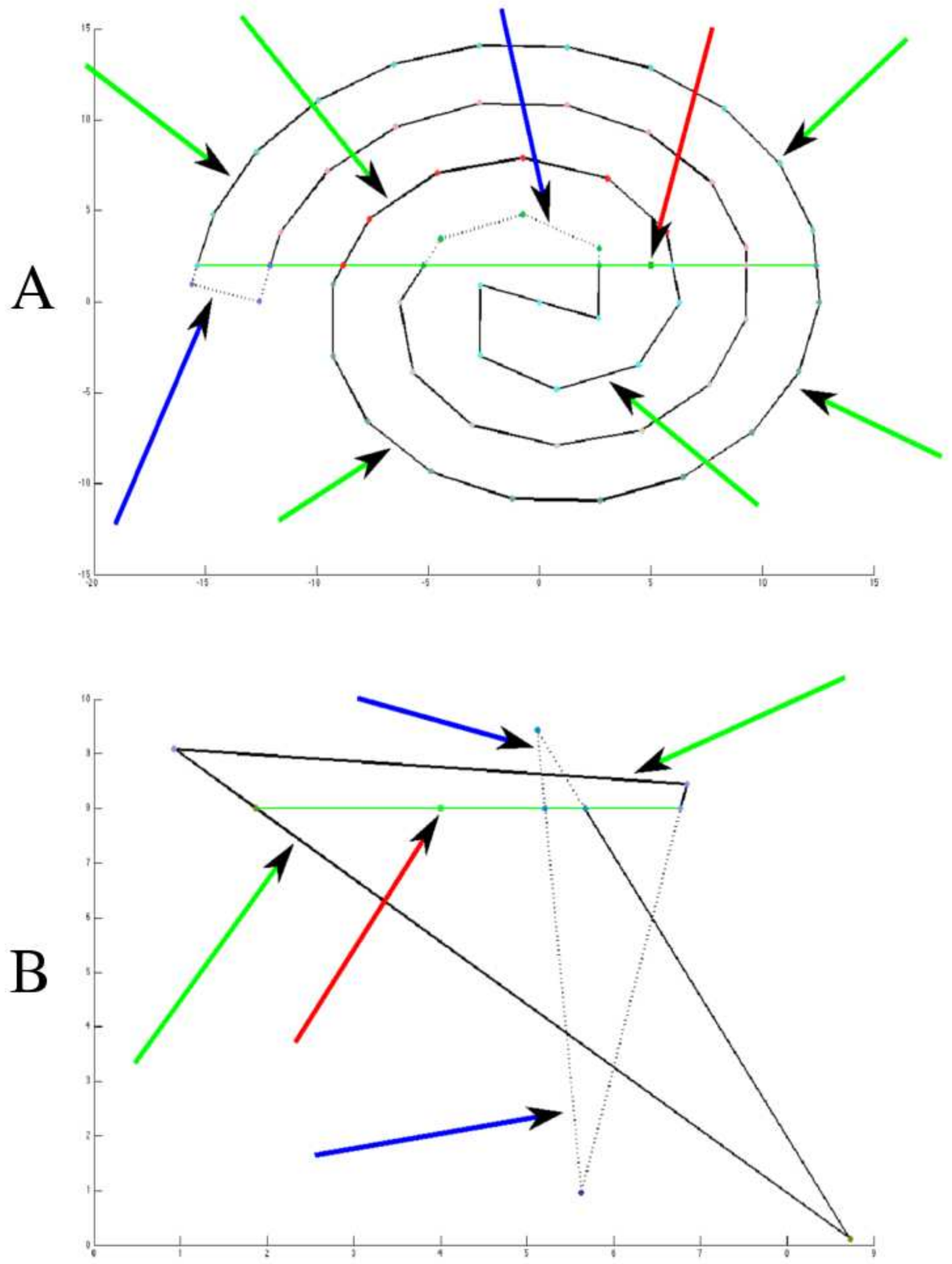}
\caption{Polygon decomposition into chains. (A) closed and (B)
  intersecting polygon decomposed into chains. Green and blue arrows
  indicate the valid and invalid chains, respectively. Location of
  query point w.r.t the polygon is shown via the red arrow.}
\label{fig:step_2_in}
\end{figure}
\begin{figure}[!t]
\centering
\includegraphics[width=8.5cm,height=10cm]{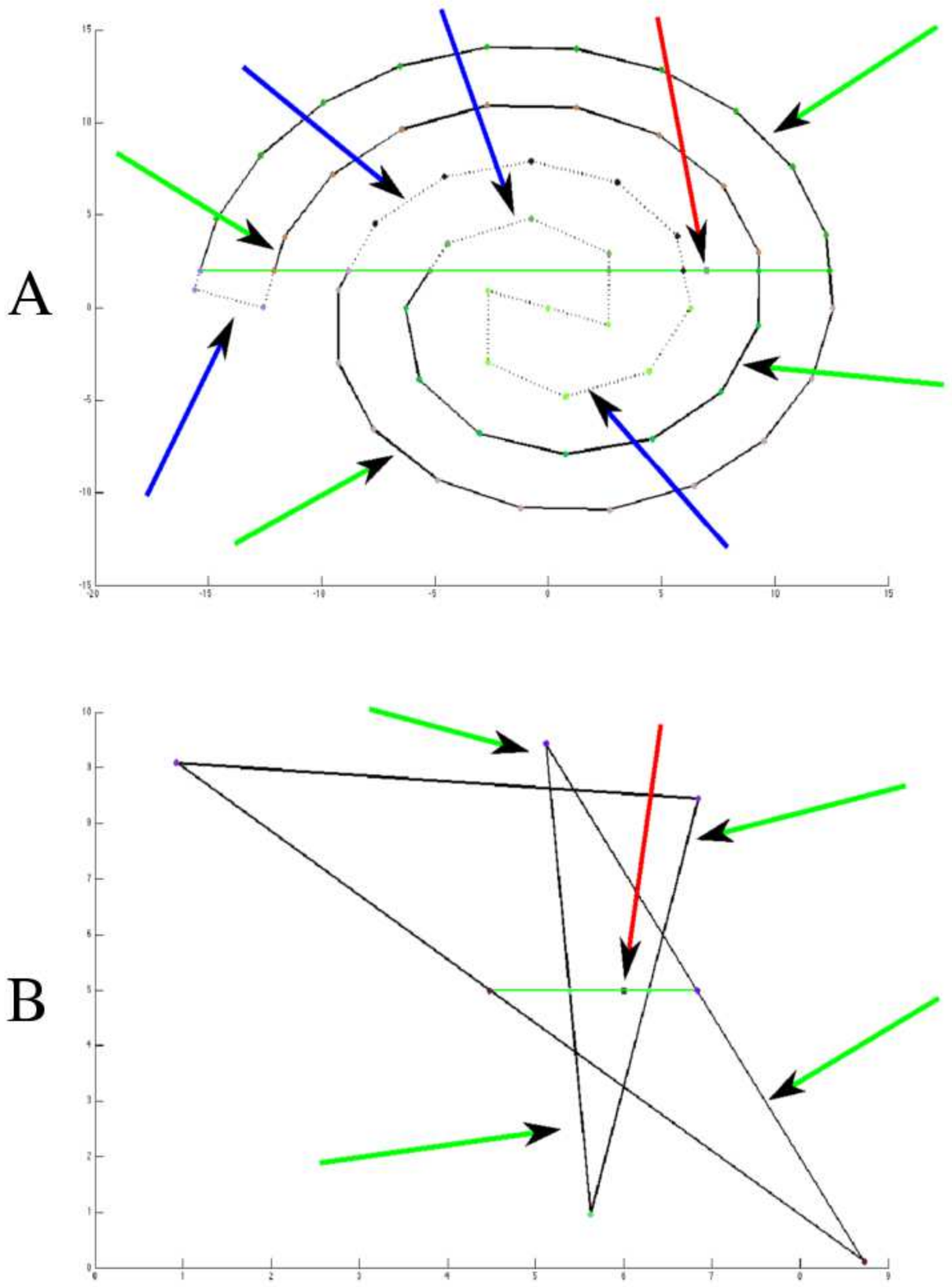}
\caption{Polygon decomposition into chains. (A) closed and (B)
  intersecting polygon decomposed into chains. Green and blue arrows
  indicate the valid and invalid chains, respectively. Location of
  query point w.r.t the polygon and is shown via the red arrow.}
\label{fig:step_2_out}
\end{figure}
%
\subsection{Chain Intersection}\label{sec:chain_intersection}
Hitherto, it is known that the $x$ coordinate of $\mathcal{S}$ lies in
the epigraph or the hypograph of the valid chains. To know that the sample
point lies within the polygon, finally it needs to be tested whether
the $y$ coordinate of $\mathcal{S}$ lies within any two nearest valid
chains. The rationale behind doing these steps will be elucidated a
little later, but before that it is important to define what the
nearest valid chains mean: \par
After the valid chains have been obtained, a similar procedure of
intersecting the chains is done using the vertical line $x =
x_{0}$. For each chain, defined by a set of vertices $(x_{i},y_{i})$,
points of intersection are computed between the edges of the chain and
the vertical line. The slope of the straight line joining two vertices
on the chain is $m_{i} = \frac{y_{i+1} - y_{i}}{x_{i+1} - x_{i}}$ and
the constant is $c_{i} = y_{i} - m_{i} x_{i}$. Now, $i$ and $i+1$ are
consecutive points on the valid chain. The value of $y$ coordinate of
the point of intersection are computed as follows: \par 
\begin{equation}
y^{int} = m_{i} x_{0} + c_{i}
\label{eq:y_int}
\end{equation}
Once the coordinate of the intersection point $(x_{0},y^{int})$ is
obtained, a test is conducted to find whether the intersection point
lies in between $(x_{i},y_{i})$ and $(x_{i+1},y_{i+1})$. This is
achieved using the affine combination property in definition
\ref{def:affinecomb}. Three cases may arise, that need to be considered: \par
\begin{enumerate}
\item $m_{i} = \pm \infty$ : If edge is a vertical line, the
  intersecting point $(x_{0},y^{int})$ is considered to lie inside the
  range of $(x_{0},y_{i})$ and $(x_{0},y_{i+1})$. This is because, if
  the chain crosses $x = x_{0}$ many times before going from left of
  $\mathcal{S}$ to the right of $\mathcal{S}$ or vice versa, then
  there can be infinitely many points on one chain that may be
  considered as intersection points. That is, it may require the
  algorithm to store many intersection points for just one chain. To
  avoid the presence of many intersection points on a single chain,
  the algorithm stores the vertex of one of the vertical edges in a
  chain.
\item $m_{i} = 0$ : If the edge is a horizontal line, the intersecting
  point $(x_{0},y^{int})$ lies on the line. This is because $x^{int} =
  x_{0}$ and $y^{int} = y_{i} = y_{i+1}$. 
\item $m_{i} \in \mathcal{R}-\{0,\pm \infty\}$ : Finally, this being
  the simplest case, it is easy to compute whether $(x_{0},y^{int})$
  lie on the line between the given points using definition
  \ref{def:affinecomb}. 
\end{enumerate}
The pictorial representation of the line $x = x_{0}$ intersecting the
valid chains are shown in figures \ref{fig:step_3_in} and
\ref{fig:step_3_out}. The green arrows indicate the points of
intersection on the valid chains that are now pointed to by the blue
arrows. The sample point is indicated via the red arrow. As mentioned
earlier, the invalid dotted chains have been removed by the algorithm,
in its final stage of processing. \par
%
\subsection{Point Inclusion Test}\label{sec:inclusion_test}
Each of the valid chains have a point of intersection with $x =
x_{0}$. Next, the algorithm sorts the chains according to the
$y^{int}$. This is done so as to arrange all the intersecting points
$(x_{0},y^{int})$ in the different chains in an ascending order, with
the sample point $(x_{0},y_{0})$ somewhere in between (if the case
is). \par
Finally, considering a pair of $y^{int}$'s at a time, i.e a pair of
chains at a time, it is tested whether $(x_{0}y_{0})$ is an affine
combination of $(x_{0},y_{i}^{int})$ and $(x_{0},y_{j}^{int})$ (for $i
< j $). If such a pair is found, then the point lies inside the
polygon, else it is outside. The only care that the algorithm takes
while executing this step is that, the pair of chains or the pair of
intersecting $y$ coordinates are mutually exclusive. This means that
if there are $n/2$ pairs, with $n$ being an even number and $y_{1}^{int}$, $y_{2}^{int}$, ...,
$y_{n}^{int}$ are the intersecting $y$ values in order, then
$(y_{1}^{int},y_{2}^{int})$, $(y_{3}^{int},y_{4}^{int})$, ...,
$(y_{n-1}^{int},y_{n}^{int})$ are mutually exclusive in the sense that
element of one pair cannot be included in any other pair. \par
The fundamental idea behind such a rule is that, in a pair, if the path of
traversal in a chain is moving from left side of $\mathcal{S}$ to
the right side (say), then the traversal in the other chain must
move from right side $\mathcal{S}$ to the left. Thus, any pair shall
not contain chains from any other pair. \par 
These ideas can be seen in figures \ref{fig:step_3_in} and
\ref{fig:step_3_out}. The green arrows mark the intersecting points on
the valid chains. The test of affine combination with respect to the sample
point $\mathcal{S}$, indicates whether the point lies within the
polygon or without. \par
\begin{figure}[!t]
\centering
\includegraphics[width=8.5cm,height=10cm]{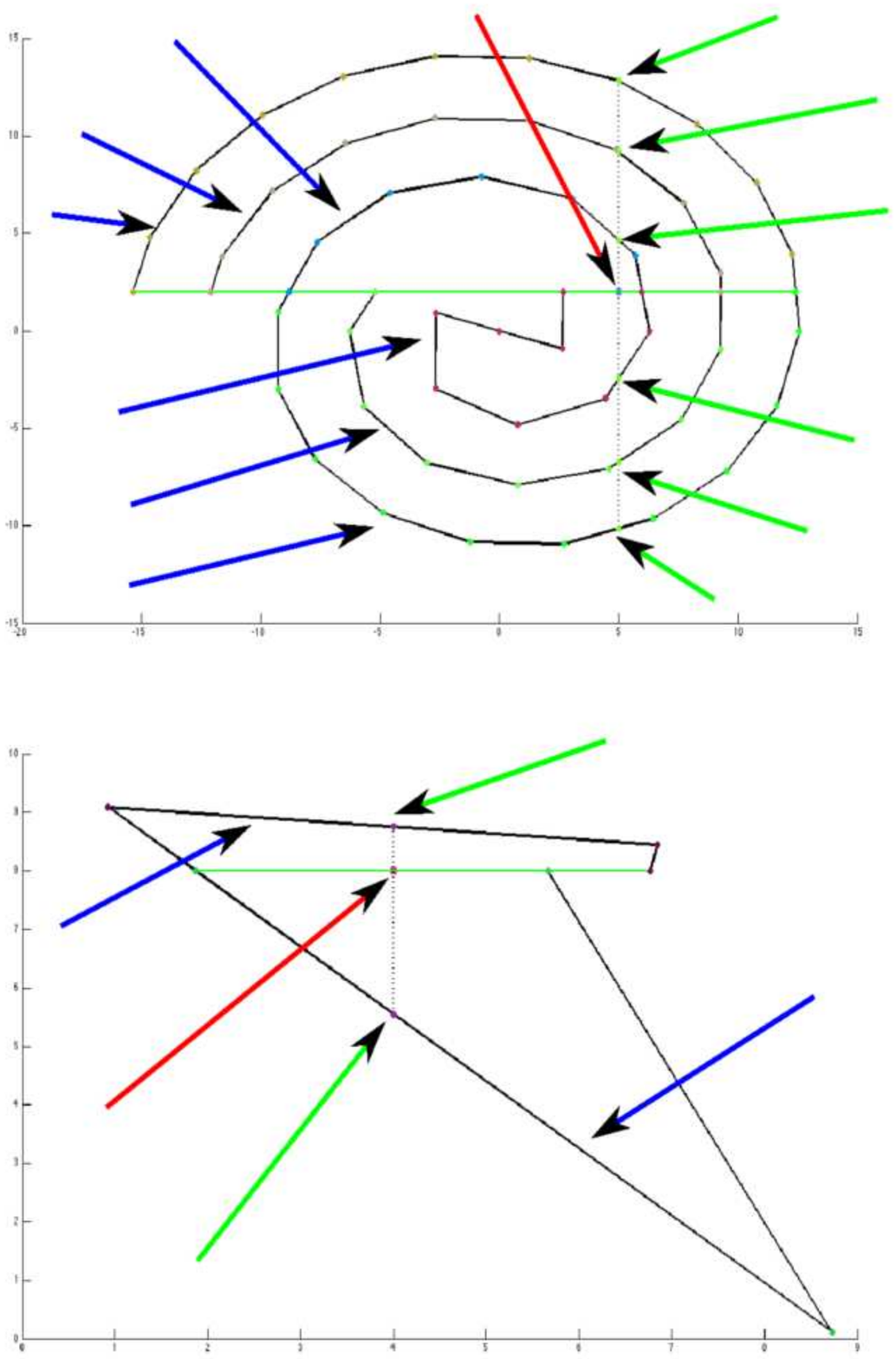}
\caption{Point inclusion test. Valid chains of (A) closed and (B)
  intersecting polygon indicated by blue arrows. Green arrows depict
  the newly found intersection points with the vertical line $x =
  x_{0}$. The red arrow shows the point lies inside the polygon.}
\label{fig:step_3_in}
\end{figure}
\begin{figure}[!t]
\centering
\includegraphics[width=8.5cm,height=10cm]{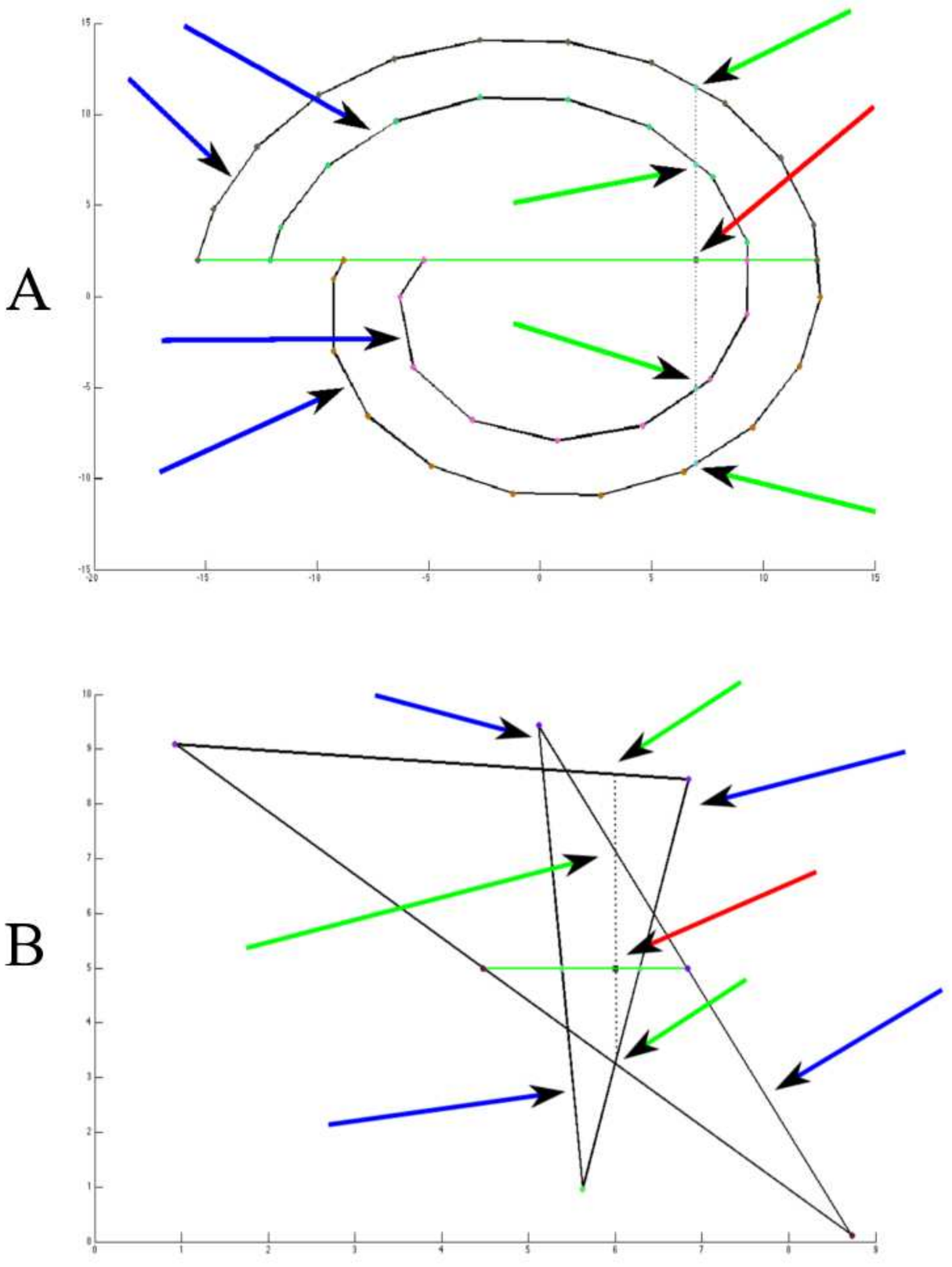}
\caption{Point inclusion test. Valid chains of (A) closed and (B)
  intersecting polygon indicated by blue and green arrows,
  respectively. Green arrows also depict the newly found intersection
  points with the vertical line $x = x_{0}$. The red arrow shows the
  point lies outside the polygon.}
\label{fig:step_3_out}
\end{figure}
%
\section{Time Complexity}\label{sec:time_complexity}
The computational complexity of the algorithm in terms of time needs
to be addressed. With $n$ edges, the algorithm takes a constant time
of $C_{1}$ to process the first step mentioned in section
\ref{sec:intersect_poly}. $C_{1}$ is the constant number of steps used
in computing the intersection of the horizontal line with each of the
edges and deciding whether the intersection point lies on the edge,
between the vertices of the edge. This process thus takes up $C_{1}n$ steps. The
newly found intersection points are added to the vertex list of the
polygon while some vertices close to the intersecting points in the
range of $\pm 10^{-5}$ or less are removed to address the floating point
precision. This execution step on the whole increases the number of
vertices by a fraction, say $f_{1}$ were $0 \leq f_{1} < 1$. \par
Thus the total number of vertices on the polygon now amount to $(1 +
f_{1})n$. The decomposition of the polygon into different chains
requires the processing of all $(1 + f_{1})n$ edges, as has been
described in section \ref{sec:decompose}. This processing is of the
order of $(1 + f_{1})n$ with some constant $C_{2}$ required for
checking if the chains are valid or not. \par
In section \ref{sec:chain_intersection}, as in section
\ref{sec:intersect_poly} the points of intersection are computed for
the vertical line and one of the edges in each of the chains. Since the
number of edges belonging to the chains is a fraction $f_{2}$ ($0
\leq f_{2} < 1$) of the total number of edges of the polygon, the
execution of this step requires $C_{1}f_{2}(1 + f_{1})n$ units of cpu
time. Finally, after sorting of the chains which requires $m log(m)$
where $m = f_{2}(1 + f_{1})n$ and checking whether the sample point is
an affine combination, requires a constant time of $C_{3}$. This
procedure is mentioned in the previous section. \par
Summing up the total time of execution, the algorithm works in
$(C_{1}n + (1 + f_{1})n + C_{2} + C_{1}f_{2}(1 + f_{1})n) +
C_{3} + m log(m) = (C_{1} + 1 + f_{1} + C_{1}f_{2} +
C_{1}f_{2}f_{1})n + C_{2} + C_{3} + m log(m)$. Let $A = (C_{1} + 1 + f_{1} + C_{1}f_{2} +
C_{1}f_{2}f_{1})$ and $B = C_{2} + C_{3}$, then the computational time
complexity of the algorithm is $An + B + m log(m) \rightarrow \mathcal{O}(n
log n)$. This time in computation is the worst case scenario, where
all the edges have to be processed in for all the four main steps
described above. \par 
\begin{figure}[!t]
\centering
\includegraphics[width=8.5cm,height=5cm]{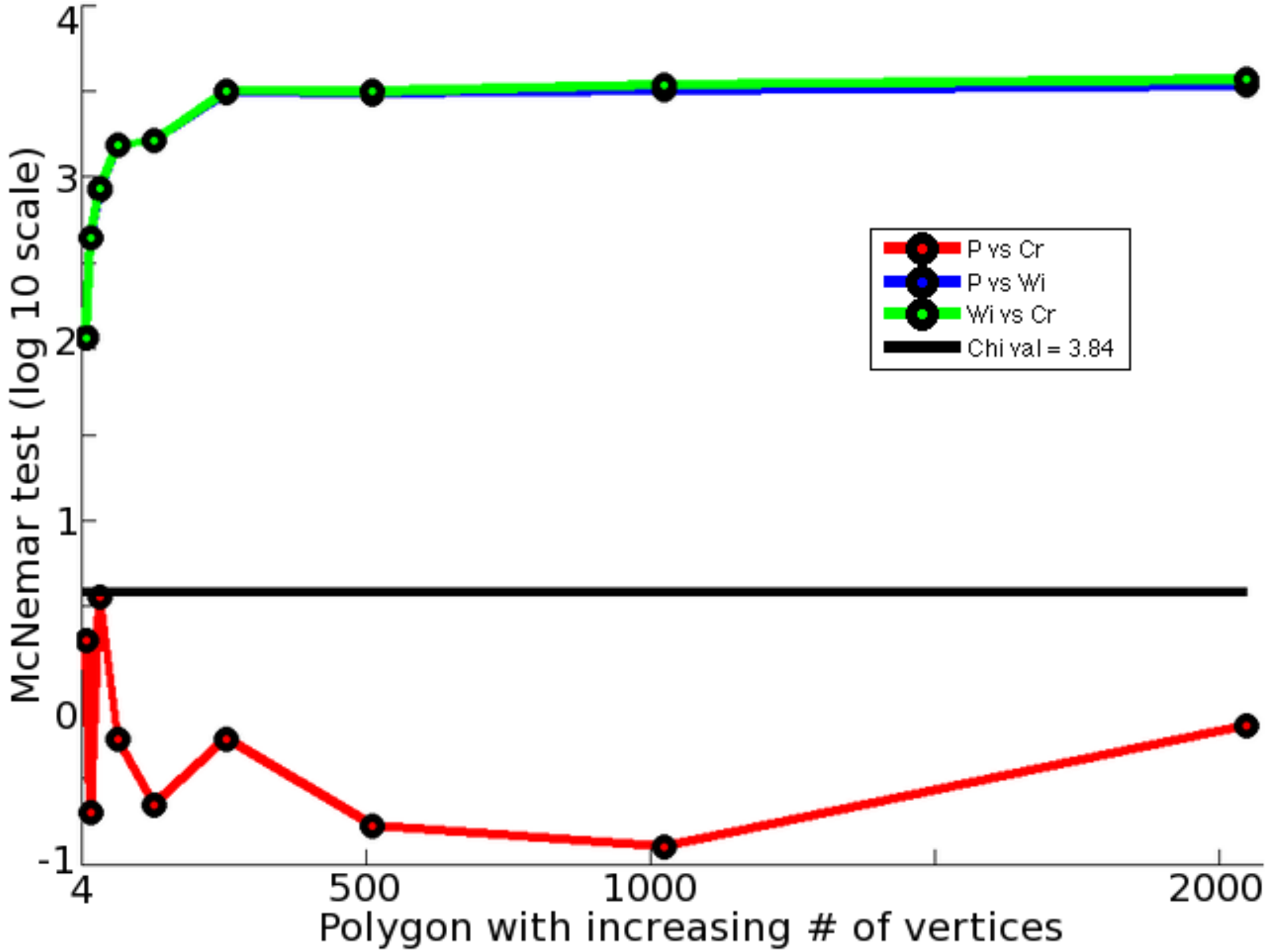}
\caption{Check of statistical significance using McNemar's Test for
  polygons with increasing vertices. Red, blue and green lines indicate the
  significance of P vs Cr, P vs Wi and Wi vs Cr, respectively. The
  black line represents the standard $\chi^{2}$ value of $3.85$ with a
  $p$ value of $0.05$. The vertical axis is the $log_{10}$
  representation of the test values.}
\label{fig:McNemartest}
\end{figure}
\section{Results}\label{sec:results}
With the proof that the proposed algorithm gives theoretically
reliable results, it would of interest to know how algorithms based on
the two widely used crossing over and winding number rule concepts,
fair. One of the measures of fairness is the test of significance of
results obtained from the algorithms. \par
Artificial polygons with number of vertices in the set $2^{i}$ where
$i \in \{2,...,11\}$ were generated. For
a particular vertex number, $10000$ polygons were generated and
stored. Each polygon was generated by randomly generating the
coordinates and joining the vertices in order. The last vertex was
finally joined with the first vertex. This created a series of
intersecting as well as non intersecting polygons. All polygons were
generated in the bounding box of $[0,1]$. \par
For each of the $10000$ polygons having a particular number of
vertices, a random test point was generated in the bounding
box. McNemar's Test \cite{McNemar:1947} was employed to check the
statistical significance of one algorithm against another. In a
$2\times2$ contingency table, the McNemar's test gives a statistic
similar to the chi squared statistic which is formulated as: \par
\begin{figure}[!t]
\centering
\includegraphics[width=8.5cm,height=5cm]{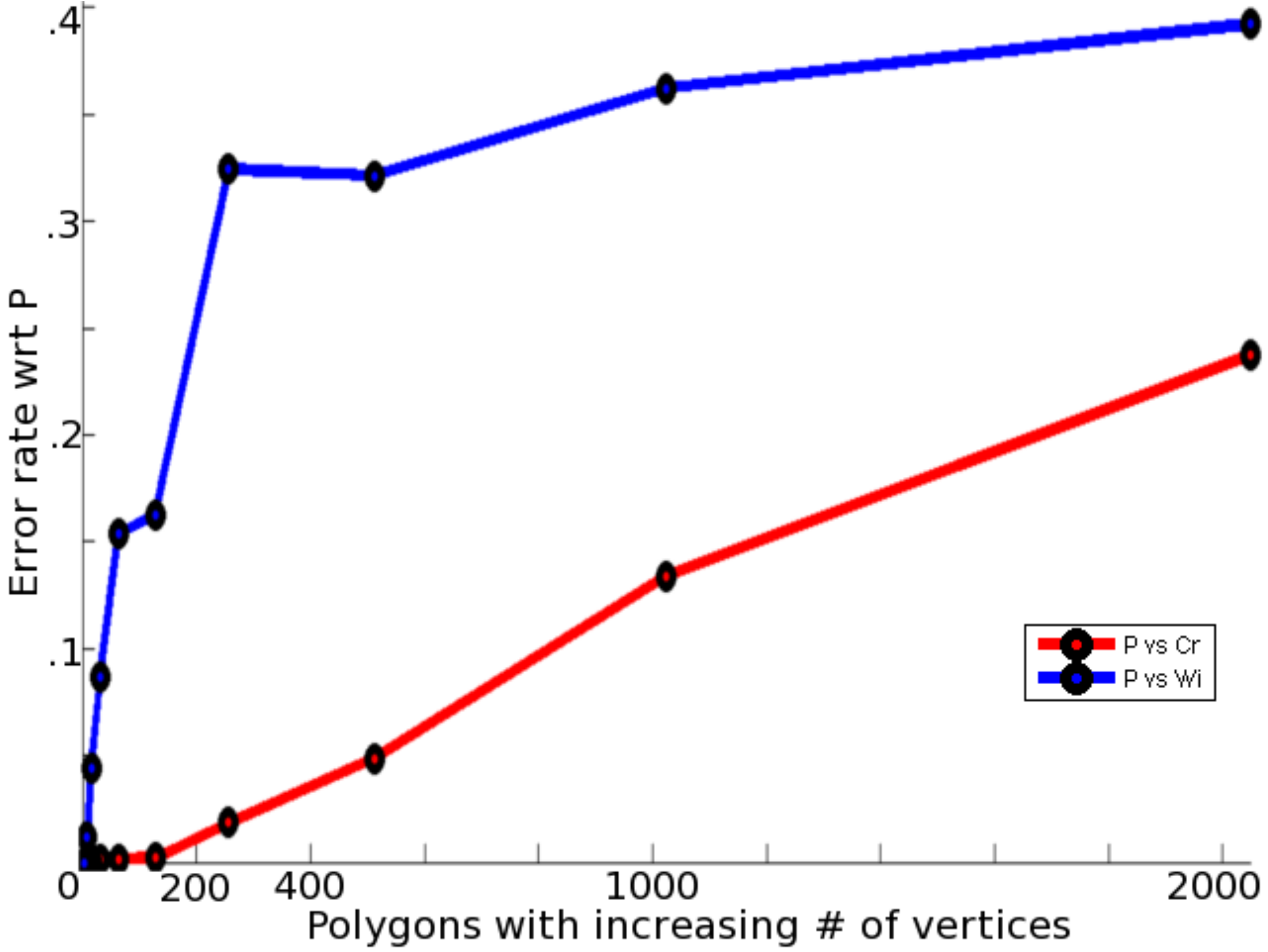}
\caption{Error rates of Cr and Wi with respect to P.}
\label{fig:Errwrtp}
\end{figure}
\begin{figure}[!t]
\centering
\includegraphics[width=8.5cm,height=5cm]{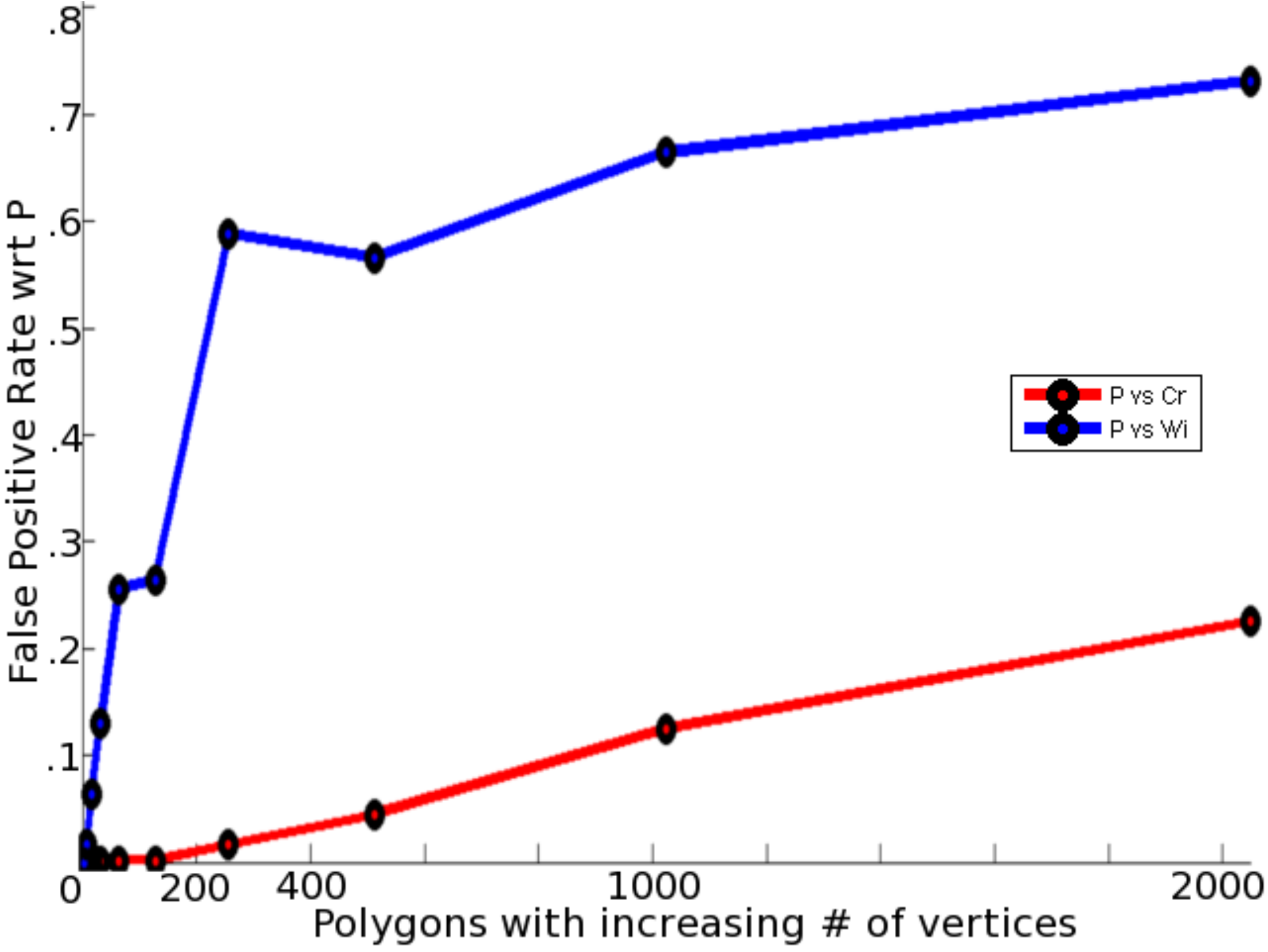}
\caption{False positive rates of Cr and Wi with respect to P.}
\label{fig:FPRwrtp}
\end{figure}
\begin{equation}
\chi^{2} = \frac{(n_{01} - n_{10})^{2}}{n_{01} + n_{10}}
\label{equ:mcnemar}
\end{equation}
\begin{equation}
\chi^{2} = \frac{(|n_{01} - n_{10}| - 1)^{2}}{n_{01} + n_{10}}
\label{equ:mcnemar_corrected}
\end{equation}
where, $n_{01}$ and $n_{10}$ are the false positives and the false
negatives, respectively. Equation \ref{equ:mcnemar_corrected} was used
as it contains the correction for discontinuity. In this
study a $\chi^{2}$ value of $3.84$ and above was set as a standard to
account for the significant difference of one algorithm against
another. Thus with $p \leq 0.05$, it is highly unlikely, that
the algorithm may be significant from the one that it is being
compared with. \par
In the present scenario, by the proof of the above theorems, it
is known that the proposed algorithm give exact results. \textbf{Thus
it is expected that the results obtained from the implementation of
the same theoretical idea shall approach ground truth with a margin of
error that is completely due to the precision format of the
computer}. Let (P) be the proposed algorithm, (Cr) the Crossing over
algorithm and (Wi) the Winding number rule algorithm. A matlab version
of the crossing over algorithm was adopted from
\cite{Darren:2007}. Matlab's \cite{Matlab} inpolygon algorithm was
taken as an implementation of the winding number rule algorithm. \par
A general analysis is presented relating to the overall behaviour of
the algorithms apropos to McNemar's Test, precision, recall, false
positive rate (Fpr) and error rate, as the number of vertices
increases. All artificially generated polygons have been stored and
can be reused with a new set of random test points to generate
similar results for checking. \par
\begin{figure}[!t]
\centering
\includegraphics[width=8.5cm,height=5cm]{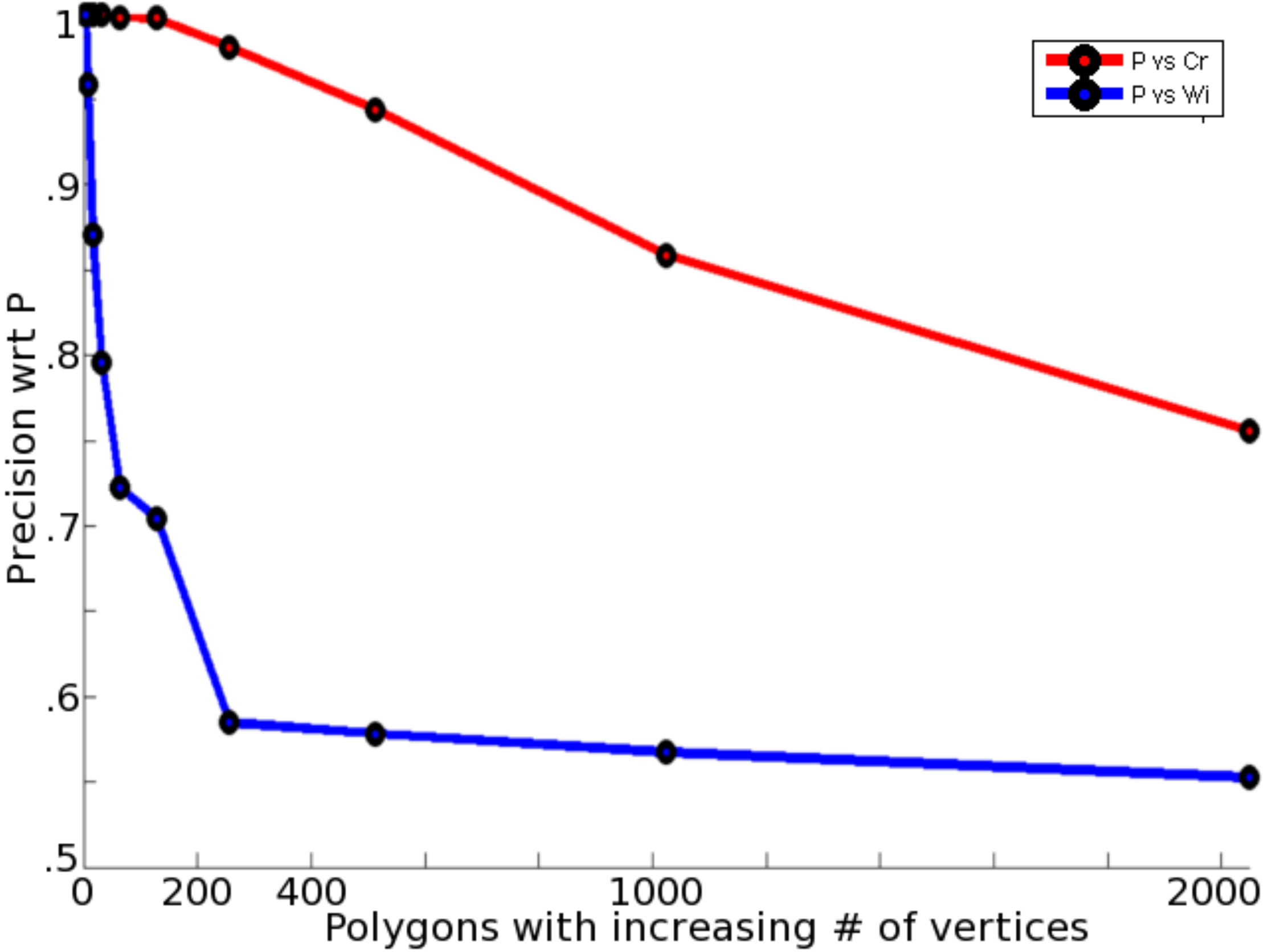}
\caption{Precision of Cr and Wi with respect to P.}
\label{fig:Precisionwrtp}
\end{figure}
\begin{figure}[!t]
\centering
\includegraphics[width=8.5cm,height=5cm]{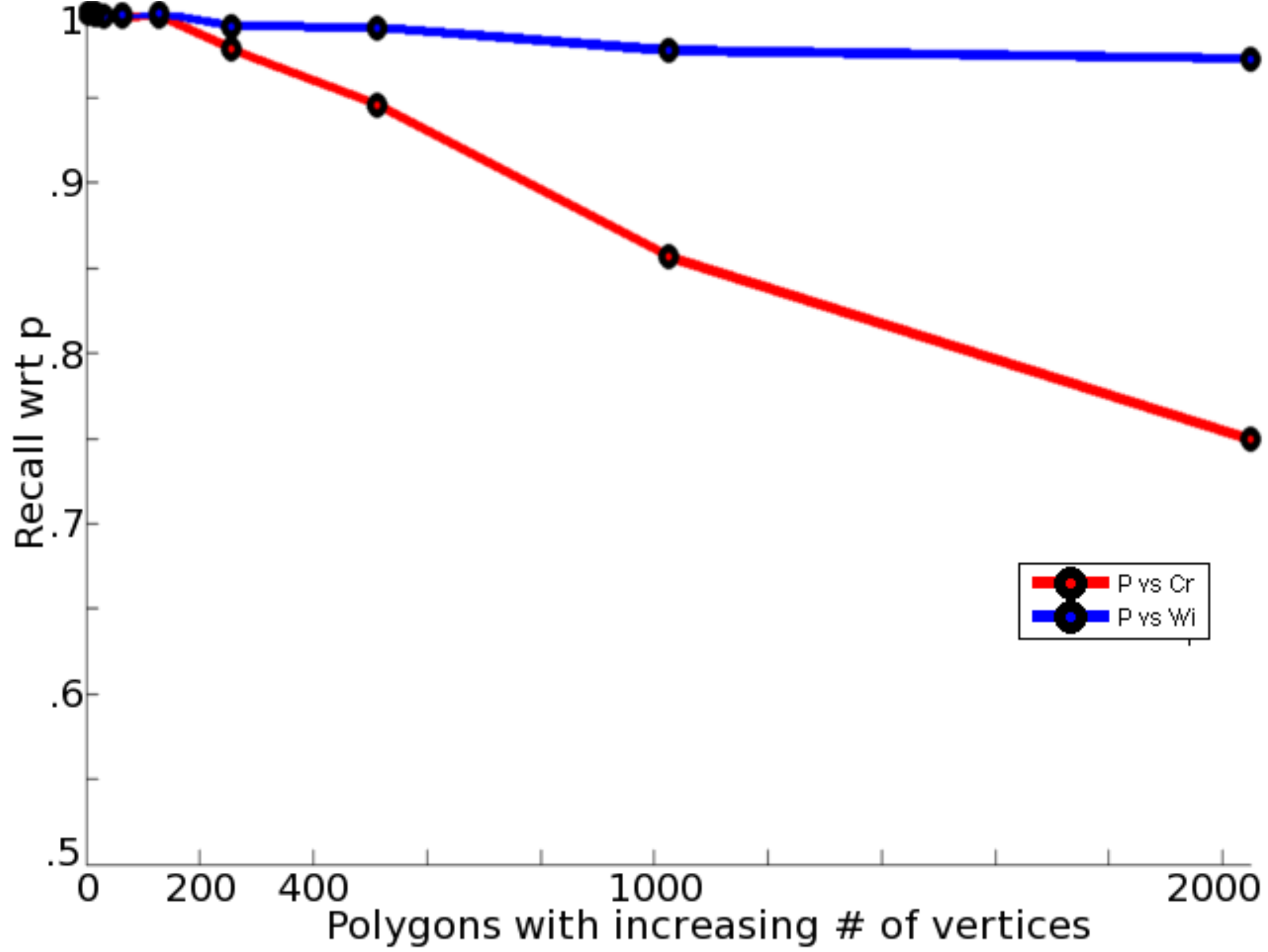}
\caption{Recall of Cr and Wi wrt P.}
\label{fig:Recallwrtp}
\end{figure}
%

\section{Discussion} \label{sec:discussion}
In order to evaluate the degree of fairness of the crossing over and
the winding number rule based algorithms w.r.t the proposed algorithm,
McNemar's test was performed on polygons with increasing number of
vertices with as sample size of $10000$. Final results in figure
\ref{fig:McNemartest} show the statistical significance in the
performance of the algorithms against each other. Clearly, it can be
seen that the results generated by Cr is significantly similar to
those generated via P (red graph in the figure). While comparing the
results of Wi with those of Cr as well as P, it can be seen that the
results are significantly different as the number of vertices
increases (blue and green graphs in the figure). \par
This apparent statistical difference is due to the fact that if the
winding number is $\ell \neq 0$, then the point lies $\ell$ times
inside the polygon. This is not true with the Cr which states that
even cross overs imply the that points are outside the polygons. P, on
the other hand sets a bound on the winding number rule, stating that
the if $\ell \neq 0$, then the point lies inside the polygon if and
only if the criterion of having a pair of chains (one above and other
below the point) is true (via proved theorems above). Thus even if
the winding number evaluates to say $2$, in case of the star polygon
for the point in figure \ref{fig:step_1_out}, P will always give the
result that the point lies outside the polygon and not twice inside the
polygon. \par
Figure \ref{fig:Errwrtp} shows the accuracy of the results obtained
by the Cr vs P and Wi vs P. The behaviour of the error rate
of Wi with respect to P and Cr, again can be attributed to the explanation
in the foregoing paragraph. Even though the results of the Cr have
been shown not to be statistically significant to that of P, with
increasing number of vertices (on a sample size of $10000$), the error
rate creeps up. The accuracy by itself does not always suffice to give
the measure of correctness and thus is aided via means of the false
positive rate, precision and recall. The graphs for the same have been
depicted in figures \ref{fig:FPRwrtp}, \ref{fig:Precisionwrtp} and
\ref{fig:Recallwrtp}. \par
The false positive rate tells how much one algorithm generated falsely
true results given that base algorithm had labeled the results as
false. Thus comparing results of Cr against P and Wi against P, the
false positive rate was generated over increasing number of vertices
and calculated on a sample size of $10000$ polygons. The graphs
suggest that both the Cr and Wi start contributing to the false
positives with respect to P as the number of vertices increases. Wi
also gives greater false positives with respect to Cr. Given the
results generated by the precision, recall and false positive error
rates it can be stated that with polygons constituting large number of
vertices, even Cr deviates from P.  
\begin{figure}[!t]
\centering
\includegraphics[width=8.5cm,height=5cm]{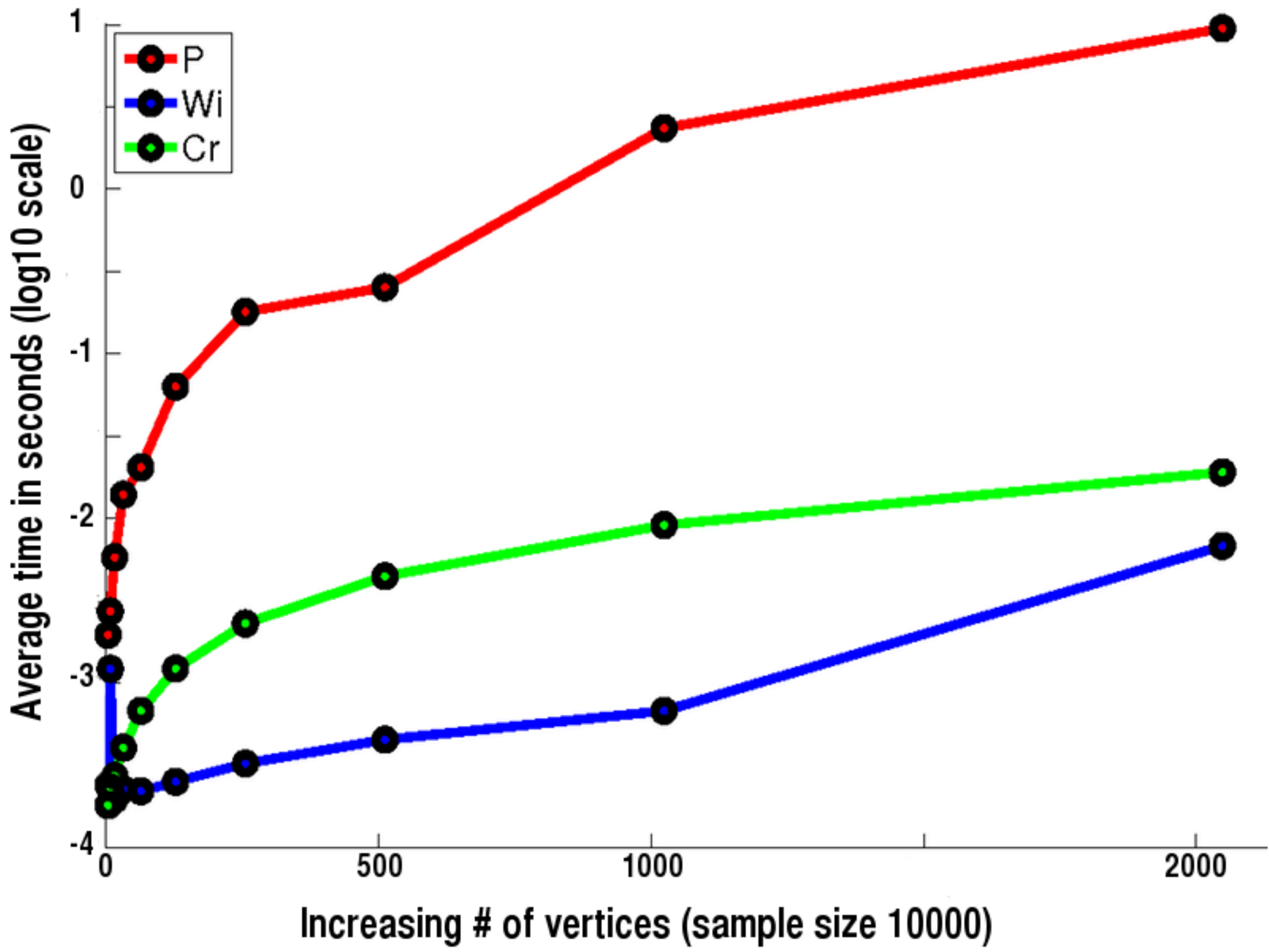}
\caption{Average time in seconds on a $log_{10}$ scale for polygons
  with increasing number of vertices. Sample size of $10000$ for each
  number of vertices.}
\label{fig:anal_01_06}
\end{figure}
The Cr which works on the fundamental idea of drawing a semi
infinite line from the sample point thus crossing the edges/vertices of
polygon can give different results depending on where the line passes
through. A different line drawn through the same sample
point may lead to a contrasting result as has been discussed in the
previous section \ref{sec:covseh}. Another case that happens in
Cr is that if the line passes through say one edge and a vertex (or
two vertices), then the results based on even odd rule, change. It has
often been pointed out that such cases are very rare and can be dealt
by infinitesimally shifting the sample point and the line
slightly. First, the the idea of being a rare case does not imply that
the problem is solved. Secondly, the solution of shifting the sample
point does work, but again it is argued that would it give
theoretically correct solution, if the sample point was on another
vertex and a slight shift would put the point outside the polygon itself. \par
With the proof of the proposed solution in the manuscript, it can be
shown that even if the polygon is intersected by two perpendicular
lines passing through the sample point at angles $\theta$ and $\theta
\pm 90$, the lines after rotation to horizontal-vertical frame, will
unarguably give correct results. Thus in rare cases also, the proposed
solution will work and give theoretical ground truths. Coming back to
the false positive rates, with a large sample, increasing size and
complexity of the polygons, the Cr and Wi are bound to give false
positives with respect to P. The precision and the recall
graphs in figures \ref{fig:Precisionwrtp} and \ref{fig:Recallwrtp},
suggest the concentration and the retrieval of the results. \par
Lastly, the average time in seconds on a log scale has been plotted in
figure \ref{fig:anal_01_06} to depict the time consumption by the
algorithms based on the concepts of the proposed solution, the
crossing over and the winding number rule. The proposed solution
apparently takes more time than the crossing over which in turn
consumes more time than the winding number. Also, in figure
\ref{fig:anal_01_07}, for polygons with the number of the vertices
$2048$, and an increasing sample size, the cumulative time in seconds
is plotted on the log scale. These graphs suggest that the reliability
of the proposed solution comes at a computational cost. But apart from
that, the correctness of the proposed solution is not affected, which
forms the kernel of the manuscript. \par
\begin{figure}[!t]
\centering
\includegraphics[width=8.5cm,height=5cm]{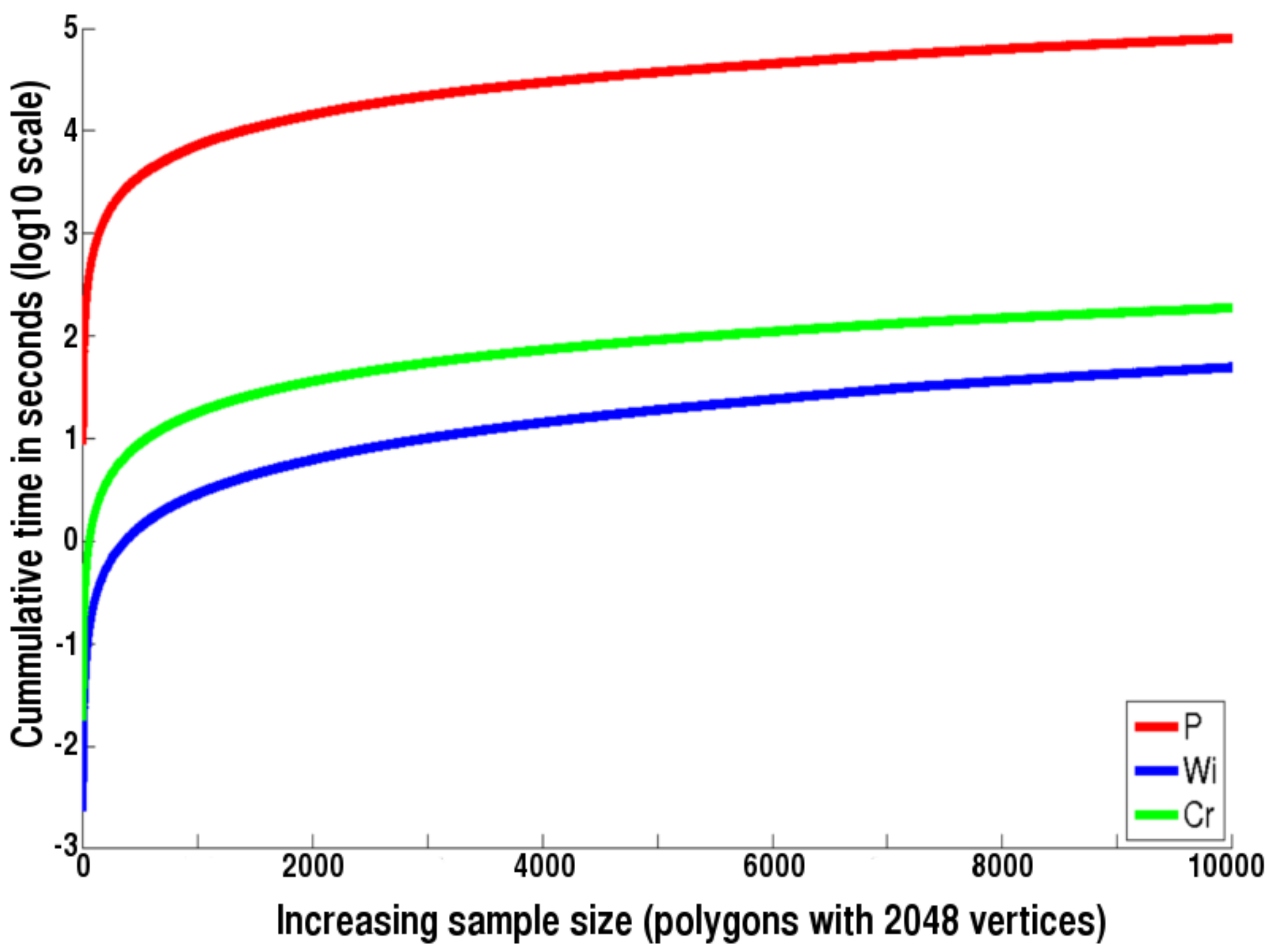}
\caption{Cumulative time in seconds on a $log_{10}$ scale for
  increasing sample size for polygons with number of vertices $2048$.}
\label{fig:anal_01_07}
\end{figure}

\bibliographystyle{ieeetr}

\end{document}